%% file: pcmi.tex
\documentclass[USenglish,a4paper,11pt]{article}
\pdfoutput=1

\input{latex-packages.sty}
\input{commands.sty}
\input{latex-style.sty}
\title{A Brief Introduction to Quantum Query Complexity}

\author{Yassine Hamoudi}
\affil{Université de Bordeaux, CNRS, LaBRI, France.\\\texttt{ys.hamoudi@gmail.com}}

\date{\today}

\begin{document}

\maketitle

\begin{abstract}
  \input{abstract.tex}
\end{abstract}


\section{Introduction}
\label{Sec:intro}
\input{introduction.tex}

\section{The Hybrid Method}
\label{Sec:hybrid}
\input{hybrid.tex}

\section{The Polynomial Method}
\label{Sec:poly}
\input{polynomial.tex}

\section{The Recording Method}
\label{Sec:record}
\input{recording.tex}

\section{The Adversary Method}
\label{Sec:adver}
\input{adversary.tex}

\section{Algorithmic Dual to the Adversary Method}
\label{Sec:dual}
\input{dual.tex}


\newpage
\printbibliography[heading=bibintoc]

\end{document}

%% file: abstract.tex
Quantum query complexity is a fundamental model for analyzing the computational power of quantum algorithms. It has played a key role in characterizing quantum speedups, from early breakthroughs such as Grover's and Simon's algorithms to more recent developments in quantum cryptography and complexity theory. This document provides a structured introduction to quantum query lower bounds, focusing on four major techniques: the hybrid method, the polynomial method, the recording method, and the adversary method. Each method is developed from first principles and illustrated through canonical problems. Additionally, the document discusses how the adversary method can be used to derive upper bounds, highlighting its dual role in quantum query complexity. The goal is to offer a self-contained exposition accessible to readers with a basic background in quantum computing, while also serving as an entry point for researchers interested in the study of quantum lower bounds.

%% file: introduction.tex
\emph{Query complexity} (also known as \emph{decision tree complexity}) is the study of algorithms that can access their input solely through an abstract operation, called a \emph{query}. Unlike other models of computation (such as circuit complexity or Turing machines), query complexity has been particularly effective in pinpointing the hardness of various computational problems. Notably, it has played a crucial role in understanding the power and limitations of quantum computing, from the early-days quantum algorithms~\cite{DJ92j,BV97j,Sim97j,Gro97j} to modern applications in cryptography (e.g., quantum random oracle model). In this respect, it is one of the few computational models in which \emph{exponential} quantum speedups have been rigorously established (e.g., Simon's problem~\cite{Sim97j}, Welded Tree~\cite{CCD03c}, Forrelation~\cite{BS21c,SSW23j}, Yamakawa-Zhandry's problem~\cite{YZ24j}).

The great success of query complexity lies, in part, in the development of \emph{lower-bound methods} that relate the minimum number of queries required to solve a problem to its combinatorial properties. Together with algorithmic upper bounds, these methods can distinguish between problems of varying complexity -- such as constant, logarithmic, cubic root, square root, etc. -- thereby establishing an \emph{unconditional} hierarchy of \emph{fine-grained} complexity classes. However, only a handful of such methods extend to the quantum setting. Indeed, many intuitive properties of classical queries do not carry over to \emph{quantum queries}, which can access the input in superposition. A foundational step in understanding the limits of quantum queries was taken in the seminal work of Bennett, Bernstein, Brassard and Vazirani~\cite{BBBV97j} through the so-called ``hybrid method''. Since then, additional methods have emerged, revealing deep connections between quantum query complexity and Boolean function analysis, matrix analysis, semidefinite programming, etc.

The goal of this document is to present four of the most prominent and successful methods for establishing quantum query lower bounds, along with some of their basic applications. The properties sought in lower-bound methods are of various types. Of course, they are expected to reduce the complexity gap with the best-known algorithms, ultimately reaching optimality. However, this can manifest in different ways. Research in query complexity has been driven by challenges such as: problems with \emph{extreme output conditions} (zero error, exponentially small success probability), query models with \emph{physical constraints} (noisy queries, bounded space, short coherence time), query models with \emph{advice} (QMA query complexity, cryptographic settings with auxiliary input), \emph{strong direct product theorems} (complexity of solving multiple instances of the same problem), \emph{composition theorems} (relating the complexity of a problem to the complexities of its components), \emph{lifting theorems} (transferring query complexity lower bounds to communication complexity), complexity of \emph{``inherently'' quantum problems} (state conversion), etc. This document primarily focuses on basic applications and does not delve into these advanced aspects, though some are supported by the methods presented in the following sections.

\fakeparagraph{Organization of the document.}
The study of quantum query complexity requires minimal prior knowledge of quantum computing. It is nevertheless assumed that the reader is familiar with the basics of quantum computing (such as Dirac notation and the circuit model), as found in any standard textbook (e.g.,~\cite{NC11b}). In \Cref{Sec:queryAlgo}, we provide a self-contained description of the computational models used in this document. \Cref{Sec:function} outlines the main problems that we will use later to illustrate the lower-bound methods.

Each of the four subsequent sections is dedicated to a different lower-bound method: hybrid (\Cref{Sec:hybrid}), polynomial (\Cref{Sec:poly}), recording (\Cref{Sec:record}) and adversary (\Cref{Sec:adver}). The hybrid, recording and adversary methods belong to the same family of techniques, and it is recommended to read them in order.

The final section (\Cref{Sec:dual}) explores a striking property of the adversary method: its ability to also provide query upper bounds (i.e., quantum algorithms) via the dual of a specific semidefinite program.

\fakeparagraph{Going further.}
Complementary introductions to quantum query complexity can be found in the survey by Buhrman and de Wolf~\cite{BdW02j}, and in the lecture notes of de Wolf~\cite{dWol19p}, Childs~\cite{Chi17p} and Ben-David~\cite{Ben20p}. A more advanced treatment of certain lower-bound techniques is provided in the Ph.D. dissertations of {\v{S}}palek~\cite{Spa06d}, Belovs~\cite{Bel14p} and Rosmanis~\cite{Ros14d}. The document will also provide pointers to the scientific literature for readers who wish to explore further.

\fakeparagraph{Acknowledgments.}
The content of this document was taught at the IAS/PCMI 2023 summer school over the course of five lectures and four problem sessions. The author wants to thank the organizers of the summer school for the invitation and the great atmosphere throughout the event. The author is also very grateful to Angelos Pelecanos for his work as a teaching assistant.


\subsection{How to Model a Quantum Query Algorithm}
\label{Sec:queryAlgo}

There are many variants of query complexity, depending on the computational power given to the algorithms, the assumptions made about the input and the conditions required for the output. In this document, we primarily focus on the simple -- yet very informative -- setup of computing Boolean functions with bounded-error algorithms, as defined next.

\begin{enumerate}[label={},itemindent=-2em,leftmargin=2em]
  \item {\bf (Boolean input alphabet)}\, The input to a query problem is an $n$-bit string $x = x_1\dots x_n \in \rn^n$, where each coordinate $x_i$ is called the \emph{query value} on \emph{query index} $i$.
  \item {\bf (Decision problems)}\, A query problem associates with each input~$x$ a unique solution~$f(x)$ specified by a Boolean function $\rf$.
  \item {\bf (Worst-case output condition)}\, An algorithm is said to \emph{compute}~$f$ if it correctly outputs~$f(x)$ with probability at least~$2/3$ for all inputs~$x$.
\end{enumerate}

The complexity of an algorithm will be measured by the number of queries made to the input~$x$. For a classical algorithm, a \emph{query} consists of revealing the value of a coordinate~$x_i$ for an index~$i$ chosen by the algorithm. For a quantum algorithm, the definition of a query is more subtle and will be given later. The query complexity of a function~$f$ is the smallest number of queries needed among all algorithms computing~$f$.

We briefly define the two models of classical query complexity -- \emph{deterministic} and \emph{randomized}~-- that are most studied in the literature. A deterministic query algorithm can be conveniently represented as a Boolean decision tree, where each node corresponds to a query and each leaf to an output (see \Cref{Fig:decTree})~-- hence the alternative name of \emph{decision tree complexity}. Since the output is deterministic, the condition for such an algorithm to compute $f$ is to always output~$f(x)$. A randomized query algorithm also has the ability to make random choices, which is equivalent to having a distribution over decision trees. In that case, the output condition allows the algorithm to output the wrong value with probability at most~$1/3$.

\begin{figure}
  \centering
  \begin{tikzpicture}[level 1/.style={sibling distance=50mm},level 2/.style={sibling distance=15mm},indexnode/.style={font=\large},queryedge/.style={font=\footnotesize}]
    \node[indexnode] {2?}
        child[indexnode] {node {1?}
            child {node {\bf 0} edge from parent [->] node[queryedge] [left] {$x_1 = 0$}}
            child[indexnode] {node {3?}
                child {node {\bf 0} edge from parent [->] node[queryedge] [left] {$x_3 = 0$}}
                child {node {\bf 1} edge from parent [->] node[queryedge] [right] {$x_3 = 1$}}
                edge from parent [->] node[queryedge] [right] {$x_1 = 1$}
            }
            edge from parent [->] node[queryedge] [left] {$x_2 = 0$}
        }
        child {node[indexnode] {3?}
                child {node {1?}
                    child {node {\bf 0} edge from parent [->] node[queryedge] [left] {$x_1 = 0$}}
                    child {node {\bf 1} edge from parent [->] node[queryedge] [right] {$x_1 = 1$}}
                    edge from parent [->] node[queryedge] [left] {$x_3 = 0$}
                }
                child {node {\bf 1} edge from parent [->] node[queryedge] [right] {$x_3 = 1$}}
            edge from parent [->] node[queryedge] [right] {$x_2 = 1$}
        };
  \end{tikzpicture}
  \caption{A deterministic algorithm with query complexity~$3$ that computes the \majf\ function over three bits ($\majf(x_1x_2x_3) = 1$ when two or three bits are equal to one).}
  \label{Fig:decTree}
\end{figure}
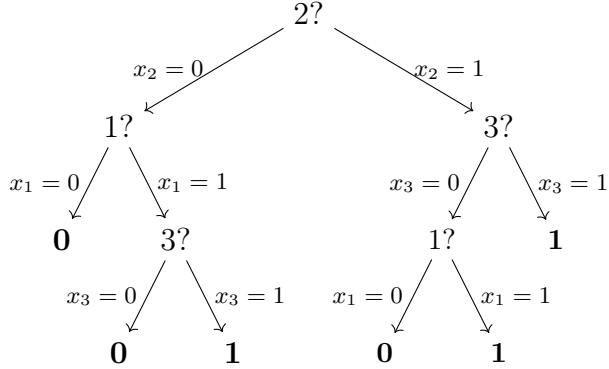

\begin{definition}[Deterministic and randomized query complexities]
  \label{Def:classQuery}
  A \emph{deterministic query algorithm} with query complexity $T$ is a rooted, ordered binary tree of depth~$T$, where each internal node is labeled with an index $i \in \ind$, and each leaf is labeled with a Boolean value. The \emph{output} of the algorithm on an input $x \in \rn^n$ is the value of the leaf obtained by starting at the root and, recursively, moving to the left child when $x_i = 0$ or to the right child when $x_i = 1$, where $i$ is the label of the current node.

  A \emph{randomized query algorithm} with query complexity $T$ is a distribution over deterministic query algorithms with query complexity $T$. The \emph{output} of the algorithm on an input $x$ is the random value obtained by sampling a deterministic query algorithm according to the distribution and evaluating it on $x$.

  The \emph{deterministic (resp. randomized) query complexity} $D(f)$ (resp. $R(f)$) of a Boolean function $\rf$ is the smallest integer $T$ such that there exists a deterministic (resp. randomized) query algorithm computing $f$ with query complexity $T$.
\end{definition}

It is important to remember that non-query operations are \emph{not counted} toward the complexity of an algorithm. Hence, a function can have low query complexity but require a large effective computation time. However, for most problems of interest, it is observed that the query complexity is a good proxy for their actual difficulty. The query complexity is nevertheless limited to be at most $n$ in the above models, since it always suffices to retrieve the entire input with~$n$ queries and apply~$f$ on it (this computation can be represented by the perfect binary tree of depth~$n$).

\begin{fact}
  For all functions $\rf$, we have $R(f) \leq D(f) \leq n$.
\end{fact}

The first question to address in defining a model of quantum query complexity is: \emph{how is the algorithm given access to the input?} If we kept the same query operator as before (by letting the algorithm observe a single coordinate at a time), then the query complexity would be \emph{identical} to that in the randomized model. Indeed, the non-query operations could simply be simulated by a classical algorithm (even if they are quantum), since their cost is not part of the query complexity anyway. Hence, the query operator must be ``truly quantum'' to make the model interesting. This is achieved by the following quantum oracles, which have the ability to query multiple coordinates in superposition.

\begin{definition}[Quantum query operators]
  Consider the Hilbert space of dimension $2n$ spanned by the vectors $\ket{i,b}$ where $i \in \ind$ and $b \in \rn$. Given an input $x \in \rn^n$, we define the two following quantum unitary operators $\ora_x$ and~$\ora^{\pm}_x$:
  \begin{center}
      \begin{tabular}{c@{\hspace{2cm}}c}
        \emph{Binary oracle} & \emph{Phase oracle} \\[5pt]
        $\ora_x \ket{i,b} = \ket{i,b \oplus x_i}$ & $\ora^{\pm}_x \ket{i,b} = (-1)^{bx_i} \ket{i,b}$
      \end{tabular}
  \end{center}
  where $\oplus$ is the XOR operation. The register holding $i$ is called the \emph{index register}, and the register holding $b$ is the \emph{value register}.
  We say that an algorithm makes a \emph{quantum query to $x$} whenever it applies $\ora_x$ or $\ora^{\pm}_x$.
\end{definition}

The first oracle $\ora_x$ is the most natural extension of the classical query operator. It writes down the value of the query in an extra register in a reversible manner, and it can act on superpositions: $\ora_x \pt[\big]{\sum_{i,b} \alpha_{i,b} \ket{i,b}} = \sum_{i,b} \alpha_{i,b} \ket{i,b \oplus x_i}$ for any amplitudes $\alpha_{i,b}$. The second oracle $\ora^{\pm}_x$ encodes the query result into the phase rather than into a quantum register, which is sometimes more convenient for use in applications. This modification makes no change to the query complexity since the two oracles can simulate each other efficiently, by conjugation with a Hadamard gate.

\begin{fact}[Equivalence between binary and phase oracles]
  \label{fact:phaseO}
  We have the following equality,
    \[\ora^{\pm}_x = \pt{\id \otimes H} \ora_x \pt{\id \otimes H}\]
  where $H$ is the Hadamard gate. In particular, any quantum algorithm that makes~$T$ quantum queries to~$x$ using~$\ora_x$ can be perfectly simulated by an algorithm that makes the same number $T$ of queries but using $\ora^{\pm}_x$ (and vice-versa).
\end{fact}

We now incorporate the quantum query operator into the definition of quantum algorithms. Since query complexity concerns only the number of query operations, we can drastically simplify the model: all operations performed between two queries (including intermediate measurements -- by the deferred measurement principle) can be grouped together into a single unitary.

In general, an algorithm may need some extra workspace beyond the $\lceil\log n \rceil+1$ qubits used for the query register $\ket{i,b}$. However, the size (or existence) of this extra memory plays no role in the lower-bound methods that will be presented in this document (in fact, it is a major research challenge to find query lower-bound methods that are also sensitive to space constraints). Hence, for the ease of notation, we will only be considering \emph{memoryless} algorithms that need no extra registers beyond $\ket{i,b}$ (an example of such an algorithm is Grover's search).

We summarize the quantum query model in the next definition. The reader can also refer to the picture in \Cref{Fig:qAlgo}, which illustrates a quantum query algorithm in the quantum circuit formalism.


\begin{figure}
  \centering
  \begin{quantikz}
    \lstick{Value register: $\ket{0}$} &[-2mm] \gate[wires=2]{U_0} \slice[style=black]{$\psx{0}$} & \gate[wires=2]{O_x} & \gate[wires=2]{U_1} \slice[style=black]{$\psx{1}$} & \ \ldots\ \qw & \gate[wires=2]{O_x} & \gate[wires=2]{U_T} \slice[style=black]{$\psx{T}$} & \meter{} \rstick{\text{Output}} \\
    \lstick{Index register: $\ket{0}$} & & \qw &  \qw & \ \ldots\ \qw & \qw & \qw & \qw
  \end{quantikz}
  \caption{Canonical form of a (memoryless) quantum query algorithm.}
  \label{Fig:qAlgo}
\end{figure}
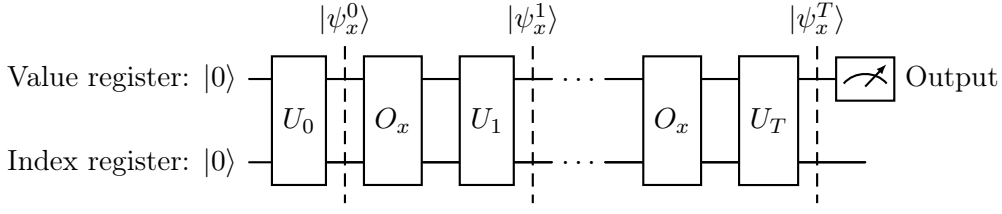

\begin{definition}[Canonical form of a memoryless quantum query algorithm]
  \label{Def:canon}
  A memoryless \emph{quantum query algorithm} with query complexity $T$ is a sequence $U_0,\dots,\allowbreak U_T$ of unitary operators acting on the Hilbert space spanned by the vectors $\ket{i,b}$ where $i \in \ind$ and $b \in \rn$.

  The \emph{intermediate state} $\psx{t}$ of the algorithm after $t \in \set{0,\dots,T}$ queries on the input $x$ is defined as,
    \[\psx{t} = U_t \ora_x U_{t-1} \ora_x \dots U_1 \ora_x U_0 \ket{0,0}.\]
  The \emph{final state} of the algorithm is $\psx{T}$ and its \emph{output} is the random bit obtained by measuring the value register of that state in the standard basis. Hence, the output is~$0$ with probability $\norm{\pt{\id \otimes \proj{0}} \psx{T}}^2$ and $1$ with probability $\norm{\pt{\id \otimes \proj{1}} \psx{T}}^2 = 1 - \norm{\pt{\id \otimes \proj{0}} \psx{T}}^2$. The \emph{success probability} $\psucc$ of the algorithm in computing a function $f$ is the smallest probability with which the algorithm outputs the correct value on any input: $\psucc = \min_{x \in \rn^n} \norm{\pt{\id \otimes \proj{f(x)}} \psx{T}}^2$. The algorithm is said to \emph{compute~$f$} if $\psucc \geq 2/3$.
\end{definition}


The most general form of query algorithms allows the unitaries $U_t$'s to act on a larger Hilbert space spanned by $\ket{i,b} \otimes \ket{w}$, where $w \in \set{1,\dots,d}$ represents an extra memory register of arbitrary dimension~$d$. The memoryless assumption amounts to restricting~$d = 1$. It is a simple exercise to verify that all the lower bound methods presented in this document extend readily to any value of~$d$.

Finally, the quantum query complexity of a Boolean function is defined as the smallest query complexity among the algorithms that compute that function.

\begin{definition}[Quantum query complexity]
  \label{Def:qc}
  The \emph{quantum query complexity} $Q(f)$ of a Boolean function $\rf$ is the smallest integer $T$ such that there exists a quantum algorithm computing $f$ with query complexity $T$.
\end{definition}

It is always possible to simulate a classical query on an index~$i$ by querying~$\ora_x \ket{i,0} = \ket{i,x_i}$ and measuring the value register. Hence, the quantum query complexity is always less than or equal to the randomized complexity (though the simulation may require considering non-memoryless quantum algorithms to carry it out).

\begin{fact}
  For all functions $\rf$, we have $Q(f) \leq R(f)$.
\end{fact}

We end this section by listing a few basic properties of the operator norm, which are used extensively in the analysis of quantum query lower bounds.

\begin{lemma}[Properties of the operator norm]
  Let $\norm{.}$ denote both the Euclidean vector norm $\norm{\ket{\psi}} = \sqrt{\sum_x \abs{\psi_x}^2}$, and the induced matrix norm (spectral norm) $\norm{A} = \max_{\psi : \norm{\ket{\psi}} = 1} \norm{A \ket{\psi}}$  over the same Hilbert space. Then, the following properties hold,
  \begin{itemize}
    \item {\it (Triangle inequality)}\, $\norm{\ps{1} + \ps{2}} \leq \norm{\ps{1}} + \norm{\ps{2}}$,
    \item {\it (Cauchy--Schwarz inequality)}\, $\abs{\ip{\psi_1}{\psi_2}} \leq \norm{\ps{1}} \cdot \norm{\ps{2}}$,
    \item {\it (Unitary invariance)}\, $\norm{U \ket{\psi}} = \norm{\ket{\psi}}$,
    \item {\it (Submultiplicativity)}\, $\norm{A \ket{\psi}} \leq \norm{A} \cdot \norm{\ket{\psi}}$,
  \end{itemize}
  for any vectors $\ket{\psi},\ps{1},\ps{2}$, unitary $U$ and matrix $A$.
\end{lemma}


\subsection{Functions of interest}
\label{Sec:function}

We list a few examples of functions that will serve as applications for the lower-bound methods later in the document. A problem is well suited for query complexity when its difficulty resides in collecting sufficient information about its input (and not, for instance, into some pre- or post-processing operations). Of course, not all problems fall into this category. However, it is often the case that part of their solution involves solving a subproblem encountered in query complexity. A good example is Shor's factoring algorithm, whose core component is a quantum query algorithm for the Period-finding problem. In that case, quantum query complexity may still be helpful in understanding whether a subroutine is optimal or not.

\fakeparagraph{\orf\ and \srcf.}
The \orf\ function returns one if and only if the input $x \in \rn^n$ contains at least one bit equal to one. This is arguably one of the most studied problem in quantum query complexity. The celebrated Grover's algorithm~\cite{Gro97j} can compute it \emph{quadratically} faster than any classical algorithm. The proof of its optimality will serve as a guiding application throughout this document. An extension of that problem -- the \srcf\ problem -- asks to locate an index~$i$ such that~$x_i = 1$ when it exists. It turns out to be no more difficult than the decision problem.

\fakeparagraph{\colf.}
The \colf\ problem asks to decide if the input contains two coordinates with the same value~$x_i = x_j$ (or to locate the indices of such coordinates). In order to make sense of this, the input alphabet must be increased beyond the Boolean domain, which we describe in \Cref{Sec:record}. The \colf\ problem has been instrumental in the development of new quantum lower-bound methods, and it is studied heavily for its applications in cryptography as well.

\fakeparagraph{\parf\ and \majf.}
The \parf\ function returns one if and only if the input $x \in \rn^n$ has an odd number of bits equal to one. It is an example of a problem for which no significant quantum speedup is possible (we will see that the query complexity decreases by only a factor of two). Another example of such a function is \majf, which returns the Boolean value that appears most frequently in the input.

\fakeparagraph{\conf.} 
The \conf\ function views the input $x \in \rn^{\binom{n}{2}}$ as the adjacency matrix of an undirected graph over~$n$ vertices, and it returns one if and only if that graph is connected. We will show how lower-bound methods have also led to the development of new quantum algorithms for this problem.

\fakeparagraph{\aotf.} 
The \aotf\ function is a composed function acting on~$n = m^2$ bits, which applies the \orf\ function to the~$m$ consecutive blocks of~$m$ input bits, followed by the \andf\ function on the results: $\andf(\orf(x_1,\dots,x_m),\allowbreak \orf(x_{m+1},\dots,x_{2m}),\dots,\orf(x_{n-m+1},\dots,x_n))$. The composition structure of this problem makes it possible to analyze its complexity in a very automated way, by simply multiplying the complexities of the \orf\ and \andf\ functions. This will result from the composition properties of the adversary method detailed in \Cref{Sec:adver,Sec:dual}.

\fakeparagraph{Total vs. partial functions.}
The functions studied in this document share a common property: their input~$x$ can take \emph{any} value in the Boolean hypercube~$\rn^n$. This necessitates restricting ourselves to problems that can be represented by functions $f$ whose domain is the entire set~$\rn^n$. Such functions are called \emph{total}, in contrast to \emph{partial} functions, which are only defined on a proper subset of $\rn^n$ to which~$x$ is \emph{promised} to belong. The lower-bound methods presented in this document can be adapted to handle partial functions as well. There is, however, a major difference between partial and total functions in terms of the achievable gaps between $Q(f)$, $R(f)$ and $D(f)$. While the gaps can be exponential for partial functions (e.g., Deutsch–Jozsa's problem~\cite{DJ92j}, Simon's problem~\cite{Sim97j}, Welded Tree~\cite{CCD03c}), they are at most polynomial for total functions. This comes from very generic lower bounds exploiting the absence of structure in the input to total functions. We state these results below for the sake of concreteness. Weaker versions of these statements will be discussed later in the document as well.

\begin{proposition}[Relations between query complexities of total functions]
  For any total function $\rf$, the gaps between its deterministic, randomized and quantum query complexities are at most,
    \[D(f) = \bo{R(f)^3}\ \textup{\cite{Nis91j}},\ R(f) = \bo{Q(f)^4} \ \textup{\cite{ABK21c}}.\]
  Moreover, there exist some functions $f,g,h : \rn^n \ra \rn$ that achieve the next gaps,
    \[D(f) = \wom{R(f)^{2}} \ \textup{\cite{ABB17j}},\ D(g) = \wom{Q(g)^{4}}\ \textup{\cite{ABB17j}},\ R(h) = \wom{Q(h)^{3}} \ \textup{\cite{BS21c,SSW23j}}.\]
\end{proposition}

These results indicate that, for total functions, the best possible speedup is at most \emph{cubic} in the randomized vs. deterministic setting and at most \emph{quartic} in the quantum vs. deterministic or randomized settings. While Grover's algorithm provides an example of a \emph{quadratic} speedup for the quantum vs. randomized setting, there are recent examples of total functions~\cite{BS21c,SSW23j} exhibiting a \emph{cubic} speedup. It is still an open question to find a \emph{super-cubic} speedup for the quantum vs. randomized setting, or to show that $R(f) = \bo{Q(f)^3}$.

%% file: hybrid.tex
The \emph{hybrid method} is the first lower-bound technique we present in this document. It originates from the work of Bennett, Bernstein, Brassard and Vazirani~\cite{BBBV97j} on the limits of quantum computers for solving NP problems. This foundational result precluded a super-quadratic speedup for the \orf\ function. Later, Grover complemented this lower bound with the celebrated quantum search algorithm~\cite{Gro97j}, showing that the optimal speedup for \orf\ is indeed quadratic.

The hybrid method gave rise to a broader family of lower-bound methods known as ``adversarial''. We will develop two other such methods later in the document (\Cref{Sec:record,Sec:adver}).


\subsection{Technique}

The hybrid method tracks the distance (equivalently, the angle) between the intermediate states of an arbitrary quantum query algorithm on two inputs $x$ and $y$ with different outcomes,~$f(x) \neq f(y)$. Initially, this distance is zero, since the algorithm always starts in the same state regardless of the input. For the algorithm to succeed, the distance must increase, approaching one so the final states are distinguishable. The challenge is to upper bound how much each query can separate the states. Intuitively, the most important indices to query are those where~$x_i \neq y_i$. The hybrid method shows this increase is governed by the \emph{total amplitude} of such indices in the query.

\begin{theorem}[Hybrid method]
  \label{Thm:hybrid}
  Consider a quantum algorithm that computes a function $\rf$ using~$T$ queries to its input. Let $\psx{t}$ denote the intermediate states of the algorithm after $t \in \set{0,\dots,T}$ queries on the input $x \in \rn^n$. Then, for any two inputs $x,y \in \rn^n$, we have:
  \begin{itemize}[leftmargin=6mm]
    \item (Initial condition) $\norm{\psx{0} - \psy{0}} = 0$,
    \item (Progress evolution) $\norm{\psx{t+1} - \psy{t+1}} \leq \norm{\psx{t} - \psy{t}} + 2 \min\limits_{z \in \set{x,y}} \norm[\big]{\pt{\sum_{i : x_i \neq y_i} \allowbreak \proj{i} \otimes \id} \psz{t}}$,
    \item (Final condition) If $f(x) \neq f(y)$ then $\norm{\psx{T} - \psy{T}} \geq 1/3$.
  \end{itemize}
\end{theorem}

\begin{proof}[Proof of the initial and final conditions]
  The initial condition is immediate, as the states are independent of the inputs: $\psx{0} = \psy{0} = U_0 \ket{0,0}$. For the final condition, by applying the Cauchy--Schwarz inequality, we have:
    \begin{align*}
      \|\psx{T} - \psy{T}\|^2
        & = 2\pt*{1-\mathrm{Re}(\ip{\psi_x^{T}}{\psi_y^{T}})} \\
        & = 2\pt[\Big]{1-\sum\nolimits_{b \in \rn}\mathrm{Re}(\braket{\psi_x^{T}}{(\id \otimes \proj{b})}{\psi_y^{T}})} \\
        & \geq 2\pt[\Big]{1-\sum\nolimits_{b \in \rn}\norm{(\id \otimes \proj{b})\psx{T}}\cdot \norm{(\id \otimes \proj{b})\psy{T}}}.
    \end{align*}
  Suppose, without loss of generality, that $f(x) = 0$ and $f(x) = 1$. Then, the correctness of the algorithm implies  $\norm{(\id \otimes \proj{0}) \psx{T}}^2 \geq 2/3$ and $\norm{(\id \otimes \proj{1}) \psy{T}}^2 \geq 2/3$. Hence, $\|\psx{T} - \psy{T}\|^2 \geq 2(1-2\sqrt{2}/3) \geq 1/9$.
\end{proof}

\begin{proof}[Proof of the progress evolution]
  The values of $x$ and $y$ are interchangeable in the proof, hence it is sufficient to consider the case where the minimum is achieved by $z = x$ in the progress evolution. By definition, the states of the algorithm at two consecutive time steps satisfy $\psx{t+1} = U_{t+1}\ora_x \psx{t}$ and $\psy{t+1} = U_{t+1}\ora_y \psy{t}$. We obtain that,
  \switchAMS{
    \begin{flalign*}
      && \norm{\psx{t+1} - \psy{t+1}}
          & = \norm{U_{t+1} (\ora_x - \ora_y)\psx{t} + U_{t+1} \ora_y(\psx{t} - \psy{t})} \\
      &&  & \leq \norm{U_{t+1} (\ora_x - \ora_y)\psx{t}} + \norm{U_{t+1} \ora_y(\psx{t} - \psy{t})} \\
      &&  & & \eqname[1cm]{triangle inequality} \\
      && & = \norm{(\ora_x - \ora_y)\psx{t}} + \norm{\psx{t} - \psy{t}} \\
      && & & \eqname{unitary invariance of the norm}.
    \end{flalign*}
    }{
    \begin{align*}
      \norm{\psx{t+1} - \psy{t+1}}
        & = \norm{U_{t+1} (\ora_x - \ora_y)\psx{t} + U_{t+1} \ora_y(\psx{t} - \psy{t})} \\
        & \leq \norm{U_{t+1} (\ora_x - \ora_y)\psx{t}} + \norm{U_{t+1} \ora_y(\psx{t} - \psy{t})} \tag{triangle inequality} \\
        & = \norm{(\ora_x - \ora_y)\psx{t}} + \norm{\psx{t} - \psy{t}} \tag{unitary invariance of the norm}.
    \end{align*}
  }
  The oracles $\ora_x$ and $\ora_y$ coincide when the query register holds an index $i$ such that $x_i = y_i$. Hence, $(\ora_x - \ora_y)\psx{t} = (\ora_x - \ora_y) \pt{\sum_{i : x_i \neq y_i} \proj{i} \otimes \id} \psx{t}$. Using that $\norm{\ora_x - \ora_y} \leq \norm{\ora_x} + \norm{\ora_y} \leq 2$ and the submultiplicativity of the norm, we conclude that $\norm{(\ora_x - \ora_y)\psx{t}} \leq 2 \norm[\big]{\pt{\sum_{i : x_i \neq y_i} \proj{i} \otimes \id} \psx{t}}$.
\end{proof}

The quantity $\norm[\big]{\pt{\proj{i} \otimes \id} \psx{t}}^2$ is sometimes called the \emph{query weight} of~$i$ in the state~$\psx{t}$, as it corresponds to the probability of measuring that index in the query register of $\psx{t}$. A direct corollary to the hybrid method states that an algorithm must assign a sufficiently large query weight to indices where $x_i \neq y_i$ in order to successfully distinguish between the two inputs.

\begin{corollary}[Query weights lower bound]
  \label{Cor:hybrid}
  Consider a quantum algorithm that computes a function $\rf$ in~$T$ queries. Let $x,y \in \rn^n$ be two inputs such that~$f(x) \neq f(y)$. Then, the query weights must satisfy the inequality
    $\sum_{t = 0}^{T-1} \sqrt{\sum_{i : x_i \neq y_i} \norm[\big]{\pt{\proj{i} \otimes \id} \psx{t}}^2} \geq 1/6$.
\end{corollary}

The quantity $\frac{1}{T} \sum_{t = 0}^{T-1} \sum_{i : x_i \neq y_i} \norm[\big]{\pt{\proj{i} \otimes \id} \psx{t}}^2$ is sometimes more convenient to manipulate, as it has an operational meaning: it represents the probability of finding a witness index~$i$ that distinguishes between $x$ and $y$ when the query register is measured at a random step of the algorithm. By combining the above corollary with the Cauchy--Schwarz inequality, this probability must be at least~$1/(36T^2)$ for the algorithm to succeed. In comparison, it can be easily shown that for randomized algorithms, this probability is at least~$\om{1/T}$, thereby enabling stronger lower bounds.

\fakeparagraph{Why is it called ``hybrid''?}
The hybrid method has another formulation that is particularly useful, for instance, in security proofs for cryptographic protocols. It states that, in an arbitrary algorithm, replacing a single unitary step~$U$ with another~$U'$ induces an error of at most $\norm{\pt{U - U'}\ket{\psi}}$, where~$\ket{\psi}$ is the state of the algorithm immediately before the perturbation occurs. One can recover \Cref{Cor:hybrid} by considering a sequence of hybrid algorithms $U_T \ora_y U_{T-1} \dots \ora_y U_t \ora_x \dots \ U_1 \ora_xU_0$ where the oracle behaves as $\ora_x$ until the $t$-th query, after which it switches to $\ora_y$. The error between two consecutive hybrids is at most $\norm{\pt{\ora_x - \ora_y} U_t \ora_x \dots \ U_1 \ora_xU_0 \allowbreak\ket{0,0}} = \norm{\pt{\ora_x - \ora_y} \psx{t}} \leq 2 \norm[\big]{\pt{\sum_{i : x_i \neq y_i} \proj{i} \otimes \id} \psx{t}}$. The interested reader can find more details in a survey by Vazirani~\cite{Vaz98j}.


\subsection{Applications}

We describe two canonical uses of the hybrid method for problems whose output is sensitive to a small perturbations of their input.


\fakeparagraph{Application 1: The \orf\ function.}
Our first application is to recover the original quantum lower bound $Q(\orf) = \om{\sqrt{n}}$ for the \orf\ function. In order to use \Cref{Thm:hybrid}, we must find a way to bound the term $\norm[\big]{\pt{\sum_{i : x_i \neq y_i} \proj{i} \otimes \id} \psx{t}}$ appearing in the progress evolution. This amounts to finding pairs of inputs $(x,y)$ with $f(x) \neq f(y)$ such that the query weight increases slowly on the bits that differ between the two inputs.

The crucial idea is to focus the analysis on the all-$0$ input $x\super{0} = (0,\dots,0)$ and the $1$-sparse inputs $y\super{1} = (1,0,\dots,0), y\super{2} = (0,1,0,\dots,0), \dots, y\super{n} = (0,\dots,0,1)$. The merit of this choice is that it simplifies the progress evolution into:
  \[\norm{\ket{\psi_{x\super{0}}^{t+1}} - \ket{\psi_{y\super{i}}^{t+1}}} \leq \norm{\ket{\psi_{x\super{0}}^{t}} - \ket{\psi_{y\super{i}}^{t}}} + 2 \norm[\big]{\pt{\proj{i} \otimes \id} \ket{\psi_{x\super{0}}^{t}}}\]
since $x\super{0}$ and $y\super{i}$ differ only on the bit at position $i$. We claim that, on average over $i$, the query weight $\norm[\big]{\pt{\proj{i} \otimes \id} \ket{\psi_{x\super{0}}^{t}}}^2$ is $1/n$. Indeed, $\sum_{i=1}^n \norm[\big]{\pt{\proj{i} \otimes \id} \ket{\psi_{x\super{0}}^{t}}}^2 = 1$ since the states $\pt{\proj{i} \otimes \id} \ket{\psi_{x\super{0}}^{t}}$ are orthogonal. Hence, it requires $T = \om{\sqrt{n}}$ queries to progress from $\norm{\ket{\psi_{x\super{0}}^{0}} - \ket{\psi_{y\super{i}}^{0}}} = 0$ (initial condition) to $\norm{\ket{\psi_{x\super{0}}^{T}} - \ket{\psi_{y\super{i}}^{T}}} \geq 1/3$ (final condition). We make this argument more formal below.

\begin{proposition}[Hybrid method applied to \orf]
  \label{Prop:hybridOR}
  The quantum query complexity of the \orf\ function is at least~$Q(\orf) \geq \sqrt{n}/6$.
\end{proposition}

\begin{proof}
  Consider any quantum algorithm that computes the \orf\ function using some number~$T$ of queries. Define the progress measure after $t \in \set{0,\dots,T}$ queries as:
    \[\Delta_t = \sum_{i=1}^n \norm{\ket{\psi_{x\super{0}}^{t}} - \ket{\psi_{y\super{i}}^{t}}}\]
  where $\ket{\psi_{x\super{0}}^{t}}$ (resp. $\ket{\psi_{y\super{i}}^{t}}$) is the state of the algorithm after $t$ queries on input~$x\super{0}$ (resp.~$y\super{i}$). By \Cref{Thm:hybrid}, the progress must be $\Delta_0 = 0$ initially and at least $\Delta_T \geq n/3$ at the end because $f(x\super{0}) \neq f(y\super{i})$ for all $i$. At each query, the progress increases by at most $\Delta_{t+1} \leq \Delta_t + 2 \sum_{i=1}^n \norm[\big]{\pt{\proj{i} \otimes \id} \ket{\psi_{x\super{0}}^{t}}} \leq \Delta_t + 2 \sqrt{n}$, since $\sum_{i=1}^n \norm[\big]{\pt{\proj{i} \otimes \id} \ket{\psi_{x\super{0}}^{t}}} \leq \sqrt{n \sum_{i=1}^n \norm[\big]{\pt{\proj{i} \otimes \id} \ket{\psi_{x\super{0}}^{t}}}^2} = \sqrt{n}$ by the Cauchy--Schwarz inequality. Hence, $n/3 \leq \Delta_T \leq 2T\sqrt{n}$ and we necessarily have~$T \geq \sqrt{n}/6$.
\end{proof}

The hybrid method is optimal (up to a constant factor) in the case of the \orf\ function, as Grover's algorithm~\cite{Gro97j} provides a matching $\bo{\sqrt{n}}$ upper bound. A careful inspection of that algorithm reveals that its query weights on the inputs $y\super{i}$ are $\norm{(\proj{i} \otimes \id )\ket{\psi_{y\super{i}}^{t}}}^2 \approx t^2/n$. Hence, it falls close to the lower bound stated in \Cref{Cor:hybrid} since $\sum_{t=0}^{T-1} \sqrt{t^2/n} = \bo{1}$ when~$T = \bo{\sqrt{n}}$.


\fakeparagraph{Application 2: Block sensitivity.}
Our second application relates the query complexity and the \emph{block sensitivity} of a Boolean function. The block sensitivity is a combinatorial measure of complexity introduced by Nisan~\cite{Nis91j} that generalizes the notion of \emph{sensitivity}. Informally, it is the largest number of disjoint bit flips in the input that can change the output of the function.

\begin{definition}[Block sensitivity]
  \label{Def:bs}
  The \emph{block sensitivity} $\bs(f)$ of a function $\rf$ is the largest integer~$s$ such that there exist an input $x \in \rn^n$ and $s$ disjoint subsets $B_1,\dots,B_s \subseteq \ind$ satisfying $f(x^{B_j}) \neq f(x)$ for all $1 \leq j \leq s$, where $x^{B_j} \in \rn^n$ is defined by $x^{B_j}_i = 1 - x_i$ when $i \in B_j$ and $x^{B_j}_i = x_i$ otherwise.
\end{definition}

One can check that the block sensitivity is maximal (i.e., equal to $n$) for the \orf, \andf\ and \parf\ functions for instance. This is proved by observing that the output value is flipped on each block $B_j = \set{j}$ when the input is, respectively, $x = 0^n$, $x = 1^n$ and $x \in \rn^n$.

An interesting example of a function with lower block sensitivity, due to Rubinstein~\cite{Rub95j}, is constructed as follows. Suppose that~$n = m^2$ is a square number and~$m$ is even. First, define the function~$g : \rn^m \ra \rn$ on~$m$ bits such that $g(x) = 1$ when all coordinates of~$x$ are zero except at two consecutive indices (i.e., $x_{2j} = x_{2j+1} = 1$ for some~$j$). Next, consider the composed function~$f = \orf \bullet g$, which applies~$g$ to the~$m$ consecutive substrings of size~$m$ in~$x \in \rn^n$ and returns the \orf\ of the results, i.e.,
  $f(x) = \orf\pt*{g(x_1,\dots,x_m),\dots,g(x_{n-m+1},\dots,x_n)}$.
The block sensitivity of~$f$ is easily shown to be~$\bs(f) = n/2$. The input achieving this optimum is the all-$0$ string, with the blocks~$B_j = \set{2j,2j+1}$ for all $j \in \set{0,\dots,n/2-1}$. If the blocks were restricted to be of size one (which results in the \emph{sensitivity} of~$f$), then the optimum would only be~$\sqrt{n}$.

The block sensitivity was identified by Nisan as characterizing the complexity of computing a function in a certain  model of parallel computation. He also described the following connection with classical query complexity.

\begin{proposition}[Lemma 4.2 in~\cite{Nis91j}]
  The randomized query complexity of any function $f : \rn^n \ra \rn$ is at least~$R(f) \geq \bs(f)/3$.
\end{proposition}

\begin{proof}
  The proof is by contradiction. Let $x$ be an input on which $f$ attains its block sensitivity, and consider $s = \bs(f)$ disjoint blocks $B_1,\dots,B_s$ with $f(x^{B_j}) \neq f(x)$ for all $j$. Suppose that the algorithm makes fewer than~$s/3$ queries. Then, on input $x$, there exists a block $B_j$ of indices that is queried with probability less than $1/3$. When the algorithm does not query that block, it must be wrong with probability at least $1/2$ either on input $x$ or $x^{B_j}$ (since it cannot distinguish between the two). Hence, the algorithm fails with probability at least $(1-1/3) \cdot 1/2 = 1/3$ on some input.
\end{proof}

The hybrid method also helps us establish a quantum lower bound in terms of the block sensitivity. Another proof of this result will be given later using the polynomial method (\Cref{Prop:bsPoly}).

\begin{proposition}[Hybrid method applied to block sensitivity]
  \label{Prop:hybridBs}
  The quantum query complexity of any function $f : \rn^n \ra \rn$ is at least~$Q(f) \geq \sqrt{\bs(f)}/6$.
\end{proposition}

\begin{proof}
  We proceed similarly to the proof of \Cref{Prop:hybridOR}. First, we identify a hard set of inputs: fix $x \in \rn^n$ and $s = \bs(f)$ disjoint subsets $B_1,\dots,B_s \subseteq \ind$ such that $f(x^{B_j}) \neq f(x)$. Next, we define the progress measure after~$t$ queries as:
    $\Delta_t = \sum_{j=1}^s \norm{\ket{\psi_{x}^{t}} - \ket{\psi_{x\super{B_j}}^{t}}}$.
  By \Cref{Thm:hybrid}, the initial and final conditions imply that $\Delta_0 = 0$ and $\Delta_T \geq s/3$. Moreover, the progress increases by at most
    \[\Delta_{t+1} \leq \Delta_t + 2 \sum_{j=1}^s \norm[\Big]{\pt[\Big]{\sum_{i \in B_j} \proj{i} \otimes \id} \psx{t}}\]
  since $x$ and $x\super{B_j}$ differ only in the indices in $B_j$. By the Cauchy--Schwarz inequality, it is at most $\sum_{j=1}^s \norm[\big]{\pt{\sum_{i \in B_j} \proj{i} \otimes \id} \psx{t}} \leq \sqrt{s \sum_{j=1}^s \norm[\big]{\pt{\sum_{i \in B_j} \proj{i} \otimes \id} \psx{t}}^2} \leq \sqrt{s}$. Hence, $s/3 \leq \Delta_T \leq 2T\sqrt{s}$.
\end{proof}

The two above propositions give the best dependence on $\bs(f)$ that one can hope for in general. This is witnessed by the \orf\ function that saturates the bounds: $R(\orf) = \ta{n} = \ta{\bs(\orf)} $ and $Q(\orf) = \ta{\sqrt{n}} = \ta{\sqrt{\bs(\orf)}}$. In general, however, the block sensitivity may not be equal to the query complexity. There is at most a cubic gap with the deterministic complexity
  \[D(f) = \bo{\bs(f)^3}\]
(see~\cite[Lemma 5.3]{BBC+01j} for an algorithm). Together with \Cref{Prop:hybridBs}, it leads to the polynomial relationship $Q(f) \leq D(f) = \bo{Q(f)^6}$ between the deterministic and quantum query complexity of any function $f$. It is a major open problem to determine whether $D(f) = \bo{\bs(f)^2}$.

%% file: polynomial.tex
The \emph{polynomial method} was introduced in a work by Beals, Buhrman, Cleve, Mosca and de Wolf~\cite{BBC+01j}. It establishes a deep connection between quantum query complexity and the approximation of real functions by low-degree polynomials. One of its great successes was the first optimal lower bound for the \colf\ problem by Aaronson and Shi~\cite{AS04j}.

This method tends to differ from the other techniques presented in this document, as it does not involve analyzing the increase of a progress measure under each query. Instead, it directly relates the existence of a $T$-query algorithm computing~$f$ to the existence of a $2T$-degree multilinear polynomial approximating~$f$. The primary difficulty lies in lower bounding the degree of such multilinear polynomials, which is purely a matter of Boolean function analysis. In \Cref{Sec:polyAppli}, we will introduce some central tools that can be used to address this challenge.

The theory of polynomial approximation is also instrumental in the design of quantum algorithms, such as in the Quantum Singular Value Transformation framework~\cite{GSLW19c} or in obtaining a converse to the polynomial method~\cite{ABP19j}.


\subsection{Technique}

We are interested in polynomials with $n$ variables that are evaluated on Boolean inputs $x_1,\dots,x_n \in \rn$. We need only consider multilinear polynomials since $x_i^2 = x_i$ when $x_i$ is Boolean. A multilinear polynomial has the following form.

\begin{definition}
  A real, \emph{multilinear polynomial} over the variables $x_1,\dots,x_n$ is a polynomial of the form
    \[p(x_1,\dots,x_n) = \sum_{S \subseteq \ind} a_S \prod_{i \in S} x_i\]
  with real coefficients $a_S \in \R$. The \emph{degree} of $p$ is the size of the largest subset $S \subseteq \ind$ with a non-zero coefficient: $\deg(p) = \max_{a_S \neq 0} \abs{S}$.
\end{definition}

A fundamental result in Boolean function analysis is that any function defined on the Boolean hypercube can be represented as a multilinear polynomial. Moreover, this representation is unique.

\begin{lemma}
  For any function $f : \rn^n \ra \R$, there exists a unique multilinear polynomial~$p_f$ over the variables $x_1,\dots,x_n$ such that $f(x) = p_f(x)$ for all $x \in \rn^n$.
\end{lemma}

\begin{proof}
  We construct a multilinear polynomial $p_f$ that coincides with $f$. The proof of uniqueness is left to the reader. We start with indicator functions. For any $y \in \rn^n$, let $1_y : \rn^n \ra \rn$ denote the indicator function that evaluates to~$1$ if and only if the input is $x = y$. The multilinear polynomial $p_y(x) = \prod_{i : y_i = 1} x_i \prod_{i : y_i = 0} (1-x_i)$ coincides with~$1_y$ on all Boolean inputs. If $f$ is an arbitrary function, it suffices to take the linear combination $p_f(x) = \sum_{y \in \rn^n} f(y) p_y(x)$.
\end{proof}

There are several measures of complexity associated with the polynomial representation of a Boolean function. It is often relevant to also consider the polynomials that are close to~$f$ in some sense. Here, we will focus on the exact degree $\deg(f) = \deg(p_f)$, and the approximate degree $\adeg(f)$, which is obtained by minimizing over all polynomials that \emph{pointwise approximate}~$f$.

\begin{definition}[Exact and approximate degrees]
  Consider a Boolean function $\rf$. The \emph{(exact) degree}~$\deg(f)$ of~$f$ is the degree of its multilinear polynomial representation: $\deg(f) = \deg(p_f)$. The \emph{approximate degree} $\adeg(f)$ of~$f$ is the smallest integer $d$ such that there exists a polynomial~$p$ of degree $d$ that approximates $f$ in the sense
    \[\abs{p(x) - f(x)} \leq 1/3\]
  for all $x \in \rn^n$.
\end{definition}

Notice that the exact degree is always at most~$n$, and the approximate degree is less than or equal to it. On the other hand, it was recently established~\cite[Theorems 1 and 18]{ABK21c} that the gap between the two quantities is at most quadratic,
  \begin{equation}
    \label{Eq:deg}
    \adeg(f) \leq \deg(f) \leq \min\set[\big]{n,9\adeg(f)^2}. 
  \end{equation}
The \orf\ function provides an example that saturates this bound (a direct proof of the lower bound $\adeg(\orf) = \om{\sqrt{n}}$ will be given in \Cref{Prop:ordeg}).

\begin{proposition}[Example 3.11 in~\cite{NS94j}]
  \label{Prop:orApprox}
  The exact and approximate degrees of the \orf\ function satisfy, respectively, $\deg(\orf) = n$ and $\adeg(\orf) = \bo{\sqrt{n}}$.
\end{proposition}

\begin{proof}
  The exact degree of \orf\ is easily obtained from the fact that its polynomial representation is $\orf(x_1,\dots,x_n) = 1 - (1-x_1)(1-x_2)\dots(1-x_n)$.

  For the approximate degree, we observe that it suffices to construct a \emph{univariate} polynomial~$q(z)$ of degree $\bo{\sqrt{n}}$ such that $q(0) \in [0,1/3]$ and $q(z) \in [2/3,1]$ for all $z \in \set{1,\dots,n}$. Indeed, the multivariate polynomial $p(x) = q(x_1+\dots+x_n)$ is an approximation of the \orf\ function of degree $\deg(p) \leq \deg(q)$. The desired polynomial $q$ can be obtained out of the $k$'th Chebyshev polynomial~$T_k(z)$ of degree $k = \bo{\sqrt{n}}$, by choosing $q(z) = aT_k(bz)+c$ for some constants $a,b,c \in \R$. The construction is detailed in~\cite{NS94j} or~\cite[Lemma 7]{BT22j}.
\end{proof}

One can expect that the functions with high degrees are harder to compute. A result in that direction is easy to establish for the deterministic and randomized query complexities.

\begin{proposition}[Classical polynomial method]
  The deterministic and randomized query complexities of any function $f : \rn^n \ra \rn$ satisfy, respectively, $D(f) \geq \deg(f)$ and $R(f) \geq \adeg(f)$.
\end{proposition}

\begin{proof}
  We first consider the case of deterministic query algorithms. By \Cref{Def:classQuery}, there exists a decision tree of depth $T = D(f)$ computing $f$. We construct a polynomial of degree at most $T$ in the variables $x_1,\dots,x_n$ that coincides with the output of the tree. The proof proceeds by induction on the depth. If the tree consists of a single leaf ($T = 0$), the algorithm makes no query and the corresponding polynomial is simply the constant output associated with the leaf. For depth $T \geq 1$, let $i$ be the index queried at the root of the tree, and let~$p_L$ (resp.~$p_R$) be the polynomial of degree at most~$T-1$ constructed recursively for the left (resp. right) subtree of the root. Then, the polynomial $p = (1-x_i)p_L + x_i p_R$ is of degree at most $T$ and $p(x) = f(x)$ for all $x \in \rn^n$. We conclude that $\deg(f) \leq \deg(p) \leq T = D(f)$ (the equality $\deg(f) = \deg(p)$ holds if $p$ is also multilinear, which happens, for instance, when the algorithm never queries the same index twice).

  Now, consider a randomized query algorithm with query complexity $T = R(f)$. Such an algorithm corresponds to a probabilistic distribution over decision trees $D$ of depth at most~$T$. Using the same construction as above, we associate with each tree $D$ a polynomial $p_D$ of degree at most $T$ that coincides with the output of that tree. We then define the polynomial $p$ as the linear combination of the polynomials $\alpha_D p_D$, where $\alpha_D \in [0,1]$ is the probability of selecting $D$ in the distribution.
  The crucial observation is that $p(x)$ corresponds to the probability that the randomized algorithm outputs $1$ on input $x \in \rn^n$. By assumption, this probability is at least $p(x) \geq 2/3$ when $f(x) = 1$, and at most $p(x) \leq 1/3$ when $f(x) = 0$. Thus,~$p$ is an approximating polynomial for~$f$. Since the approximate degree is the smallest degree of an approximating polynomial for~$f$, we conclude that $\adeg(f) \leq \deg(p) \leq T = R(f)$.
\end{proof}

The example of the \orf\ function shows that the lower bound on $R(f)$ in terms of the approximate degree is not necessarily tight, since $\adeg(\orf) = \bo{\sqrt{n}}$ but $R(\orf) = \om{n}$. This comes as no surprise since the crux of the polynomial method will be to show that $\adeg(f)/2$ is in fact a lower bound on the quantum query complexity. The proof exploits the following fundamental property, stating that the amplitudes of a quantum state after one query are univariate polynomials in the input bits $x_1,\dots,x_n$.

\begin{lemma}
  \label{Lem:poly}
  For any input $x \in \rn^n$, the effect of the quantum oracle operator on a basis state is
    \[\ora_x\ket{i,b} = (1-x_i) \ket{i,b} + x_i \ket{i,b\oplus 1}\]
  for all $i \in \ind$ and $b \in \rn$.
\end{lemma}

\begin{proof}
  This follows directly from the definition: $\ora_x\ket{i,b} = \ket{i,b\oplus x_i}$.
\end{proof}

By applying the lemma repeatedly, one obtains that the amplitudes after $T$ queries are multilinear polynomials of degree $T$. Since the probability of outputting~$1$ is the \emph{square} of an amplitude, it follows that any algorithm can be transformed into an approximating polynomial for~$f$ whose degree is twice the number of quantum queries.

\begin{theorem}[Polynomial method]
  \label{Thm:polynomial}
  Consider a quantum algorithm that computes a function $\rf$ using~$T$ queries to its input. Let $p(x) \in [0,1]$ denote the probability that the algorithm outputs $1$ on the input $x \in \rn^n$. Then,
  \begin{itemize}
    \item (Approximation) $\abs{p(x) - f(x)} \leq 1/3$ for all $x \in \rn^n$,
    \item (Polynomial degree) $\deg(p) \leq 2T$.
  \end{itemize}
  In particular, the approximate degree of $f$ is at most $\adeg(f) \leq 2T$.
\end{theorem}

\begin{proof}
  The first point is immediate, as the probability of outputting $1$ must be $p(x) \geq 2/3$ when $f(x) = 1$ and $p(x) \leq 1/3$ when $f(x) = 0$.

  For the second point, let $\psx{t}$ denote the intermediate states of the algorithm after $t \in \set{0,\dots,T}$ queries on the input $x \in \rn^n$. By definition, the probability of outputting $1$ is
    \[p(x) = \norm{(\id \otimes \proj{1}) \psx{T}}^2 = \sum_{1\leq i \leq n} \abs{\ip{i,1}{\psi_{x}^T}}^2.\]
  It is sufficient to show that, for all $i \in \ind$ and $b \in \rn$, the complex-valued function
    $x \mapsto \ip{i,b}{\psi_{x}^t}$
  is a polynomial in $x$ of degree at most $t$. We prove it by induction on $t$. The base case is trivial since $\psx{0} = U_0\ket{0,0}$ and $\ip{i,b}{\psi_{x}^0}$ is independent of $x$. Suppose that the statement holds for $t$. Then, for $t+1$, the inner product is $\ip{i,b}{\psi_{x}^{t+1}} = \braket{i,b}{U_{t+1}\ora_x}{\psi_{x}^{t}}$. Let $U^{\dagger}_{t+1}\ket{i,b} = \sum_{1 \leq j \leq n,c \in \rn} \alpha_{j,c}^{\dagger} \ket{j,c}$ denote the decomposition of the state $U^{\dagger}_{t+1}\ket{i,b}$ in the standard basis, where the complex numbers $\alpha_{j,c}$ are \emph{independent} of $x$. By \Cref{Lem:poly}, the inner product is
    \[\ip{i,b}{\psi_{x}^{t+1}} = \sum_{1 \leq j \leq n,c \in \rn} \alpha_{j,c} \pt*{(1-x_j)\ip{j,c}{\psi_{x}^t} + x_j \ip{j,c\oplus 1}{\psi_{x}^t}}.\]
  By the induction hypothesis, each term $\ip{j,c}{\psi_{x}^t}$ and $\ip{j,c\oplus 1}{\psi_{x}^t}$ are multivariate polynomials in $x$ of degree at most $t$. Hence, $\ip{i,b}{\psi_{x}^{t+1}}$ is of degree at most~$t+1$.
\end{proof}

This result has several interesting consequences in approximation theory and quantum query lower bounds. For instance, the existence of Grover's algorithm provides an alternative proof of \Cref{Prop:orApprox} on the approximate degree of the \orf\ function. These types of ``algorithmically-inspired'' polynomials are discussed in more detail in~\cite[Section 4]{BT22j}.
Another crucial implication in query complexity is that no algorithm can compute a function~$f$ using fewer than~$\adeg(f)/2$ queries, as this would result in an approximating polynomial of degree less than~$\adeg(f)$.

\begin{corollary}
  \label{Cor:poly}
  The quantum query complexity of any function $f : \rn^n \ra \rn$ is at least~$Q(f) \geq \adeg(f)/2$.
\end{corollary}

Together with \Cref{Eq:deg}, it also implies that $Q(f) \geq \sqrt{\deg(f)}/6$. The latter inequality yields, for instance, that~$Q(\orf),Q(\parf) \geq \sqrt{n}/6$ (which is not optimal in case of the \parf\ function). We will develop more direct methods for analyzing the approximate degree in the next section.


  \subsection{Applications}
  \label{Sec:polyAppli}

We present three applications of the polynomial method, each exploiting the multilinear polynomial constructed in \Cref{Thm:polynomial} in different ways.


\fakeparagraph{Application 1: Distinguishing distributions.}
We begin with a simple application of the polynomial method, although in a problem setting that slightly differs from what we have considered so far. Instead of computing a Boolean function $f$, we consider the problem of distinguishing between two probability distributions~$\dis_0$ and~$\dis_1$.

\begin{definition}[Distinguishing problem]
  Let $\dis_0$ and $\dis_1$ be two distributions over the set~$\rn^n$. In the \emph{$(\dis_0,\dis_1)$-distinguishing problem}, the algorithm is given oracle access to an input $x \in \rn^n$ sampled either from $\dis_0$ or $\dis_1$, and it must output a bit $b \in \rn$ that maximizes the \emph{distinguishing advantage}:
    \[\abs[\big]{\pr{b = 1 \given x \sim \dis_0} - \pr{b = 1 \given x \sim \dis_1}}.\]
  We say that two distributions are \emph{indistinguishable} by a class of algorithms if no algorithm in that class can make the advantage nonzero. Distributions that are indistinguishable from the uniform distribution $\unif$ are called \emph{pseudorandom}.
\end{definition}

The construction of pseudorandom distributions is a topic of central importance in cryptography and in complexity theory.
We show that the following \emph{$k$-wise independence} condition is sufficient for a distribution to be pseudorandom against quantum algorithms that make few queries.

\begin{definition}[$k$-wise independence]
  We say that a distribution $\dis$ over $\rn^n$ is \emph{$k$-wise independent} if for all subset $S \subseteq \ind$ of size at most~$k$, its marginal distribution on the coordinates indexed by~$S$ is uniform: $\pr_{x \sim \dis}{x_i = a_i, \forall i \in S} = 2^{-\abs{S}}$ for all choices of~$a_i \in \rn$.
\end{definition}

It is easy to see that $k$-wise independent distributions are pseudorandom for classical algorithms that make at most~$k$ queries, since the knowledge of any~$k$ input bits provides no information on whether the distribution is uniform or not.
Using the polynomial method, we demonstrate a similar result for quantum algorithms, although with a slightly stronger assumption on the number of queries.

\begin{proposition}[Polynomial method applied to the distinguishing problem]
  Let~$\dis$ be a $2k$-wise independent distribution over $\rn^n$. Then,~$\dis$ is pseudorandom for the class of quantum quantum algorithms that make at most~$k$ queries.
\end{proposition}

\begin{proof}
  Fix any quantum algorithm that outputs a bit~$b$ after making~$k$ queries to its input~$x$. By \Cref{Thm:polynomial}, the probability~$p(x)$ that it outputs~$b = 1$ on input~$x$ is given by a function $p : \rn^n \ra [0,1]$ of degree at most $\deg(p) \leq 2k$. Let $\set{a_S}_{S \subseteq \ind, \abs{S} \leq 2k}$ denote the coefficients of the multilinear polynomial that coincides with it, $p(x) = \sum_{S \subseteq \ind, \abs{S} \leq 2k} a_S \prod_{i \in S} x_i$. The distinguishing advantage of the algorithm is,
  \switchAMS{
    \begin{align*}
      \abs[\big]{\pr{b = 1 \given x \sim \dis} - \pr{b = 1 \given x \sim \unif}}
       & = \abs[\big]{\ex_{x \sim \dis}{p(x)} - \ex_{x \sim \unif}{p(x)}} \\
       & = \Big|\sum_{S : \abs{S} \leq 2k} a_S \Big(\ex_{x \sim \dis}[\Big]{\prod_{i \in S} x_i} \\
       & \qquad\qquad\qquad\qquad - \ex_{x \sim \unif}[\Big]{\prod_{i \in S} x_i}\Big)\Big|.
    \end{align*}
  }{
   \begin{align*}
     \abs[\big]{\pr{b = 1 \given x \sim \dis} - \pr{b = 1 \given x \sim \unif}}
      & = \abs[\big]{\ex_{x \sim \dis}{p(x)} - \ex_{x \sim \unif}{p(x)}} \\
      & = \abs*{\sum_{S : \abs{S} \leq 2k} a_S \pt*{\ex_{x \sim \dis}[\Big]{\prod_{i \in S} x_i} - \ex_{x \sim \unif}[\Big]{\prod_{i \in S} x_i}}}.
   \end{align*}
  }
  Since $\dis$ is $2k$-wise independent, its marginal distribution over any set~$S$ of at most~$2k$ indices is uniform.
  It implies that the product of any~$2k$ coordinates has the same expected value under the distributions $\dis$ and $\unif$.
  Hence, $\ex_{x \sim \dis}{\prod_{i \in S} x_i} = \ex_{x \sim \unif}{\prod_{i \in S} x_i}$ and the distinguishing advantage of the algorithm is zero.
\end{proof}

This type of application of the polynomial method can be generalized to other problems relevant in cryptography, such as polynomial interpolation~\cite{KK11j}.


\fakeparagraph{Application 2: Symmetrization.}
The polynomial method leads to the study of approximating polynomials whose number $n$ of variables grows with the input size of the problem. Symmetrization refers to a family of techniques that exploit the symmetries of the problem to construct another approximating polynomial with far fewer variables -- ideally one or two -- without increasing the degree. The resulting polynomial typically inherits certain fluctuation properties from the original polynomial (such as the number of roots, large derivative, etc.), which can be exploited to lower bound its degree.

We describe the most standard symmetrization technique, due to Minsky and Papert~\cite{MP69b}, which consists of averaging the polynomial over all inputs with the same Hamming weight. This approach works well for problems that are invariant under permutations of the input bits.

\begin{proposition}[Minsky-Papert symmetrization]
  \label{Prop:mp}
  Let $p(x_1,\dots,x_n)$ be a multilinear polynomial over the variables $x_1,\dots,x_n$. Then, there exists a univariate polynomial $\psym(k)$ of degree at most $\deg(\psym) \leq \deg(p)$ such that, for all integers $k \in \set{0,\dots,n}$,
    \[\psym(k) = \ex_{x \in \rn^n : \abs{x} = k}{p(x)}.\]
\end{proposition}

\begin{proof}
  It is sufficient to prove the existence of $\psym$ when $p$ consists of a single monomial~$p(x) = \prod_{i \in S} x_i$ (the general case follows by linearity of expectation). Let~$d = \abs{S}$ and consider the univariate polynomial $\psym(k) = \frac{k(k-1) \cdots (k-d+1)}{n(n-1)\cdots (n-d+1)}$ of degree $d$. It is a simple calculation to check that,
  \[\psym(k) =
      \left\{\begin{array}{c@{\hspace{2mm}}c@{\hspace{2mm}}ll}
         0 & = & \ex_{x \in \rn^n : \abs{x} = k}{\prod_{i \in S} x_i} & \text{when $k \in \set{0,\dots,d-1}$,} \\[3mm]
         \frac{\binom{n-d}{k-d}}{\binom{n}{k}} & = & \ex_{x \in \rn^n : \abs{x} = k}{\prod_{i \in S} x_i} & \text{when $k \in \set{d,\dots,n}$.}
       \end{array}
     \right.\]
\end{proof}

Notice that the dependence on~$n$ has shifted from the number of variables in $p$ to the size of the domain used to characterize~$\psym$. This often simplifies the degree analysis. We describe three applications of the Minsky-Papert symmetrization.

The first application is to show that the approximate degree of the \parf\ function is~$n$, which is the largest possible value and must be equal to the exact degree (indeed, $\parf(x_1,\dots,x_n) = (1-(1-2x_1)\cdots(1-2x_n))/2$). By \Cref{Cor:poly}, this implies that the quantum query complexity is at least $Q(\parf) \geq n/2$. This is quadratically better than the lower bound $Q(\parf) \geq \sqrt{n}/6$ obtained with the hybrid method (\Cref{Prop:hybridBs}). It is also optimal since Deutsch's algorithm~\cite{Deu85j} can compute the parity of two bits in one quantum query (the algorithm makes no error, and hence the parity of~$n$ bits can be obtained by repeating the algorithm on~$n/2$ different pairs of bits when~$n$ is even).

\begin{proposition}[Symmetrization method applied to \parf]
  The approximate degree of the \parf\ function is equal to $\adeg(\parf) = n$. Hence, the quantum query complexity is at least~$Q(\parf) \geq n/2$.
\end{proposition}

\begin{proof}
  Consider any multilinear polynomial $p(x_1,\dots,x_n)$ that approximates the \parf\ function. The proof consists of showing that, necessarily, $\deg(p) \geq \deg(\psym) \geq n$. The first inequality is immediate by \Cref{Prop:mp}. The second inequality uses the property that,
    \[\abs{\psym(k)} \leq 1/3 \ \ \mathrm{for} \ k \in \set{0,2,4,\dots} \quad\mathrm{and}\quad \abs{\psym(k) - 1} \leq 1/3 \ \ \mathrm{for} \ k \in \set{1,3,5,\dots}\]
  since $\parf$ evaluates to $0$ on inputs with an even Hamming weight, and to $1$ on inputs with an odd Hamming weight. A possible representation of the polynomial~$\psym$ on the interval~$[0,n]$ is shown in \Cref{Fig:poly}. The above property implies that the polynomial $1-2\psym$ changes sign at least~$n$ times over the interval $[0,n]$. Hence, its number of roots must be at least~$n$ and~$\deg(\psym) \geq n$.
\end{proof}

\begin{figure}
  \centering
  \includegraphics[width=0.8\textwidth,trim=5cm 2.5cm 1.5cm 4.5cm,]{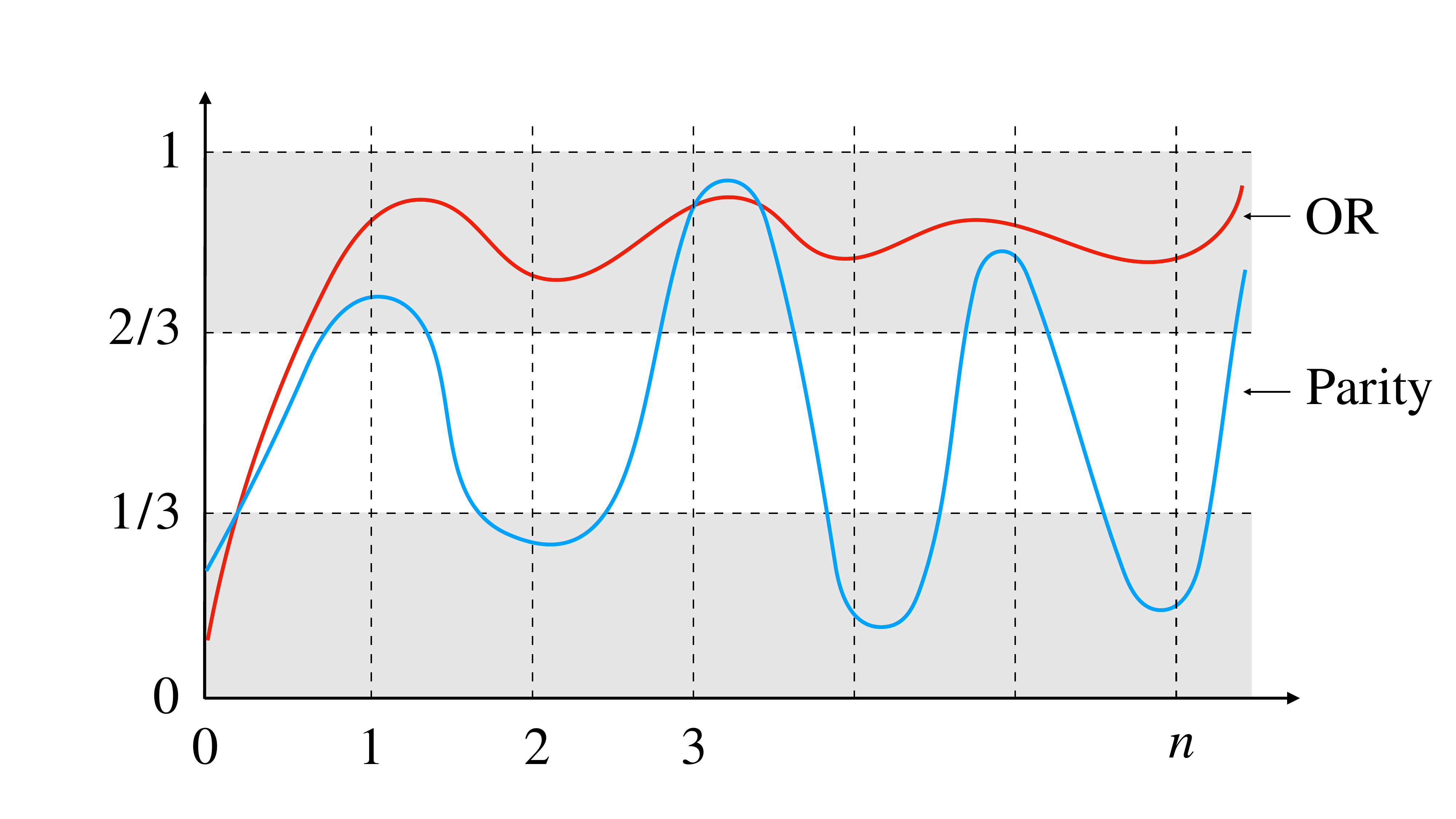}
  \caption{Possible representations of symmetrized polynomials~$\psym$ obeying the constraints for the \orf\ and \parf\ functions.}
  \label{Fig:poly}
\end{figure}

The next application provides another proof that the quantum query complexity of~\orf\ is~$\om{\sqrt{n}}$. The symmetrized polynomial $\psym$ is slightly harder to analyze here, as it requires using an inequality from approximation theory.

\begin{proposition}[Symmetrization method applied to \orf]
  \label{Prop:ordeg}
  The approximate degree of the \orf\ function is at least~$\adeg(\orf) \geq \sqrt{n/6}$. Hence, the quantum query complexity is at least~$Q(\orf) \geq \sqrt{n/24}$.
\end{proposition}

\begin{proof}
  Consider any multilinear polynomial $p$ that approximates the \orf\ function. The polynomial $\psym(k)$ derived from \Cref{Prop:mp} has the property that,
    \[\abs{\psym(0)} \leq 1/3\quad \mathrm{and}\quad  \abs{\psym(k) - 1} \leq 1/3 \ \ \mathrm{for} \ k \in \set{1,\dots,n}\]
  since $\orf$ evaluates to $0$ on inputs of Hamming weight $\abs{x} = 0$, and to $1$ on inputs of Hamming weight $\abs{x} \in \set{1,\dots,n}$. A possible representation of the polynomial~$\psym$ on the interval~$[0,n]$ is shown in \Cref{Fig:poly}. By the Mean Value Theorem, there must exist a real $x \in [0,1]$ such that the derivative at~$x$ of the function~$\psym$ is at least $\psym'(x) \geq 1/3$. We evoke the following lower bound on the degree of polynomials with such large derivatives.
  \begin{lemma}[Ehlich-Zeller and Rivlin-Cheney Theorem]
    Let $a,b,c \in \R$ (with $a < b$ and $c > 0$), $n \in \N$ and $p : \R \ra \R$ be a polynomial such that $p(k) \in [a,b]$ for all integers $k \in \set{0,1,\dots,n}$ and $|p'(x)| \geq c$ for some real $x \in [0,n]$. Then, $\deg(p) \geq \sqrt{cn/(c + b-a)}$.
  \end{lemma}
  A direct application of this lemma to the polynomial~$\psym$ with $a = -1/3$, $b = 4/3$ and $c = 1/3$ leads to the conclusion that $\deg(\psym) \geq \sqrt{n/6}$.
\end{proof}

The next application extends the previous result to a general lower bound on the approximate degree in terms of the block sensitivity (\Cref{Def:bs}).

\begin{proposition}[Symmetrization method applied to block sensitivity]
  \label{Prop:bsPoly}
  The approximate degree of any function $f : \rn^n \ra \rn$ is at least~$\adeg(f) \geq \sqrt{\bs(f)/6}$. Hence, the quantum query complexity is at least~$Q(f) \geq \sqrt{\bs(f)/24}$.
\end{proposition}

\begin{proof}
  Let $x \in \rn^n$ be an input on which $f$ attains its block sensitivity. Suppose, without loss of generality, that $f(x) = 0$ and fix $s = \bs(f)$ disjoint blocks $B_1,\dots,B_s$ with $f(x^{B_j}) = 1$ for all~$j$. We define the function $\pi : \rn^s \ra \rn^n$ that associates with each $y \in \rn^s$ the bitstring obtained by flipping all blocks in~$x$ indexed by~$y$, i.e., $\pi(y)_i = 1 - x_i$ if $i \in B_j$ and $y_j = 1$, and $\pi(y)_i = x_i$ if $i \in B_j$ and $y_j = 0$.

  Given any polynomial~$p$ that approximates~$f$, we can construct another multilinear polynomial $q(y_1,\dots,y_s)$, over~$s$ variables, such that
    \[q(y_1,\dots,y_s) = p(\pi(y))\]
  for all $y \in \rn^s$ and $\deg(q) \leq \deg(p)$ (it suffices to replace the $i$-th variable in $p$ with $x_i(1-y_j)+(1-x_i)y_j$ if $i \in B_j$). In particular,~$q(y)$ evaluates to~$p(x)$ on the all-$0$ string, and to~$p(x^{B_j})$ on the string~$y$ with a single~$1$ at position~$j$. Hence, its symmetrized polynomial must satisfy,
    \[\abs{q_{\mathrm{sym}}(0)} \leq 1/3,\quad
    \abs{q_{\mathrm{sym}}(1) - 1} \leq 1/3,\quad
    q_{\mathrm{sym}}(k) \in [-1/3,4/3] \ \ \mathrm{for} \ k \in \set{2,\dots,s}.\]
  The Ehlich-Zeller and Rivlin-Cheney Theorem yields that $\deg(q_{\mathrm{sym}}) \geq \sqrt{s/6} = \sqrt{\bs(f)/6}$.
\end{proof}

The reader interested in problems that are more challenging to symmetrize can continue her reading with the lower bounds for the \colf~\cite{AS04j}, and \aotf~\cite{Kre21j} problems.


\fakeparagraph{Application 3: Dual polynomials.}
We conclude with a more recent and powerful technique for lower bounding the approximate degree.
For convenience, in this section, we express the domain and range of Boolean functions as,
  \[f : \rnn^n \ra \rnn.\]
This amounts to replacing the bit $b \in \rn$ with $1-2b$, which is an operation that preserves the (approximate) degree.

The central idea of the dual polynomial method is to view the approximate degree as being given by the following pair of primal-dual linear programs.

\switchAMS{
  \begin{gather*}
    \textit{Primal linear program}  \\
    \boxed{\begin{array}{lll}
      \min_{\eps,p}  &     \eps    \\[2mm]
      \mbox{s.t.} & \abs{p(x) - f(x)} \leq \eps & \forall x \in \rnn^n\\[2mm]
      & \deg(p) < d &
     \end{array}} \\[5mm]
     \textit{Dual linear program} \\
     \boxed{\begin{array}{ll@{\hspace{0mm}}l}
      \max_{\phi} &    \sum_{x \in \rnn^n} \phi(x) \cdot f(x) & \\[2mm]
      \mbox{s.t.} & \sum_x |\phi(x)| = 1& \\[2mm]
      & \sum_x \phi(x) \cdot p(x) = 0 & \forall p, \deg(p) < d
      \end{array}}
  \end{gather*}
}{
\[\begin{array}{@{}cc@{}}
    \textit{Primal linear program} & \textit{Dual linear program} \\
    \boxed{\begin{array}{lll}
      \min_{\eps,p}  &     \eps    \\[2mm]
      \mbox{s.t.} & \abs{p(x) - f(x)} \leq \eps & \forall x \in \rnn^n\\[2mm]
      & \deg(p) < d &
     \end{array}}
     &
     \boxed{\begin{array}{ll@{\hspace{0mm}}l}
        \max_{\phi} &    \sum_{x \in \rnn^n} \phi(x) \cdot f(x) & \\[2mm]
        \mbox{s.t.} & \sum_x |\phi(x)| = 1& \\[2mm]
        & \sum_x \phi(x) \cdot p(x) = 0 & \forall p, \deg(p) < d
        \end{array}}
\end{array}\]
}

The variables of the primal program are the approximation parameter~$\eps$ and the $\binom{n}{< d}$ coefficients needed to represent a polynomial~$p : \rnn^n \ra \rnn$ of degree less than~$d$. The variables of the dual program are the~$2^n$ values~$\phi(x)$ needed to specify a function~$\phi : \rnn^n \ra \R$.

It is straightforward to relate the primal program with the approximate degree of~$f$: for a fixed value of~$d$, the approximate degree is \emph{at most}~$\adeg(f) < d$ if and only if the optimal value is at most~$\eps \leq 1/3$. The dual program is more interesting to interpret: by weak duality, the approximate degree is \emph{at least}~$\adeg(f) \geq d$ if one can identify a so-called \emph{dual polynomial}~$\phi : \rnn^n \ra \R$ satisfying the next conditions.

\begin{proposition}[Method of Dual Polynomials]
  \label{Prop:dualpoly}
  Let $f : \rnn^n \ra \rnn$ be a Boolean function and $d \in \set{0,\dots,n}$ be an integer. Suppose that there exists a real-valued function $\phi : \rnn^n \ra \R$ such that,
  \begin{itemize}
    \item (Correlation) $\sum_{x \in \rnn^n} \phi(x) \cdot f(x) > 1/3$,
    \item (Normalization) $\sum_{x \in \rnn^n} |\phi(x)| = 1$,
    \item (Pure high degree) $\sum_{x \in \rnn^n} \phi(x) \cdot p(x) = 0$ for all polynomials $p : \rnn^n \ra \rnn$ of degree at most $\deg(p) < d$.
  \end{itemize}
  Then the approximate degree of~$f$ must be at least~$\adeg(f) \geq d$.
\end{proposition}

In words,~$2^n \phi$ must not deviate significantly from~$f$ (correlation and normalization), while having no monomial of degree less than~$d$ (pure high degree). It suffices to exhibit one such function~$\phi$ to certify that the approximate degree of~$f$ is at least~$d$. A simple case of application is for \parf.

\begin{proposition}[Dual polynomial method applied to \parf]
  The approximate degree of the \parf\ function is $\adeg(\parf) = n$.
\end{proposition}

\begin{proof}
  The polynomial representation of the \parf\ function over $\rnn$ is given by the degree-$n$ monomial $\parf(x_1,\dots,x_n) = x_1 \cdots x_n$. It is easy to check that $\phi(x) = \frac{1}{2^n} x_1 \cdots x_n$ is a valid dual polynomial: $\sum_x \phi(x) \cdot f(x) = 1$, $\sum_x \abs{\phi(x)} = 1$ and $\deg(\phi) = n$.
\end{proof}

Unfortunately, we will not provide additional examples of dual polynomials, as they quickly become complicated to construct (even for the \orf\ function). The reader is invited to consult the survey by Bun and Thaler~\cite{BT22j} for more details and intuition on this method.

%% file: recording.tex
The \emph{recording method} (also called the \emph{compressed oracle} technique) was introduced in a recent work by Zhandry~\cite{Zha19c}. It was originally intended for security proofs in the quantum random oracle model (QROM), where the input $x \in \set{0,\dots,m-1}^n$ is often non-Boolean ($m > 2$) and is meant to represent a random hash functions $H(i) = x_i$. Since then, it has also proven useful in query lower bounds, although there are some limitations to the type of problems it applies to.

The recording method is best suited for problems that are hard \emph{on average} when the input is drawn from some sufficiently \emph{unstructured} distribution. It also requires the problem to be based on a \emph{local property}, in the sense that the output is determined by a small subset of input coordinates satisfying a certain predicate. Some of these limitations have been partially lifted since then, but we will not touch upon these improvements here. Our presentation follows the approach in~\cite{HM21c}.

A typical case of application, which will be developed in the applications section, is the problem of finding a coordinate equals to~$x_i = 1$ in a random input $x$ (with alphabet size $m \approx n$). This is a variant of the \orf\ problem called (preimage) \srcf, also solved with~$\bo{\sqrt{n}}$ queries by Grover's algorithm. Here, the input is drawn from the uniform distribution, so each coordinate is independently a solution with probability~$1/m$.


\subsection{Technique}

The recording method starts by placing a certain hard distribution~$\dis$ on the input $x$. The most well-understood case is when the coordinates~$x_i$ are \emph{independent and identically distributed} under~$\dis$. To handle some problems of interest, it is best to extend the input alphabet beyond the Boolean domain.
Hence, we will detail the method when the input is drawn from the distribution,
  \begin{center}
    $\unif:$ uniform distribution on $\set{0,\dots,m-1}^n = \Sigma^n$
  \end{center}
for some integer $m \geq 2$. The quantum query framework stated in~\Cref{Sec:queryAlgo} must first be extended in three directions:

\begin{enumerate}[label={},itemindent=-2em,leftmargin=2em]
  \item {\bf (Large input alphabet)}\, The algorithm is given access to an input $x \in \set{0,\dots,m-1}^n$ through the oracle $\ora_x \ket{i,b} = \ket{i,b + x_i \bmod m}$, where $i \in \ind$ and $b,x_i \in \set{0,\dots,m-1}$. Equivalently, it is given access to the phase oracle $\ora^{\pm}_x \ket{i,b} = \omega^{b x_i} \ket{i,b}$, where $\omega = e^{2\mathbf{i}\pi/m}$ is the $m$-th root of unity.
  \item {\bf (Relational problems)}\, The problem can have multiple valid solutions: each input $x$ is associated with a set $\rel_x \subseteq \ind$ of solutions, and the algorithm succeeds if it outputs any element from that set.
  \item {\bf (Average-case output condition)}\, The algorithm must output a valid solution with probability at least~$2/3$, where the randomness is both over the actions of the algorithm (e.g., the outcome of a measurement) \emph{and} the distribution~$\dis$ of the input.
\end{enumerate}

We make a few comments on the above definitions. First, it will be more convenient to use the phase oracle $\ora^{\pm}_x \ket{i,b} = \omega^{b x_i} \ket{i,b}$ in the recording method. Its equivalence with the oracle $\ora_x$ follows by the same arguments as in \Cref{fact:phaseO}.
Next, the second item offers more flexibility in the type of problems that can be considered, as the solution need not be Boolean, nor unique. The set of valid solutions is now represented by a \emph{relation} $\rel = \set{(x,i) : x \in \set{0,\dots,m-1}^n, i \in \rel_x}$, as opposed to a function~$f(x)$ (which would restrict the relation to singletons $\rel_x = \set{f(x)}$). For simplifying the model, we assume that the solutions can be encoded as numbers from the set $\ind$, but other sets can work as well (the \colf\ problem, defined in the applications section, would require encoding the solutions over $\ind \times \ind$).
Finally, the average-case condition is also a relaxation of the output condition stated in \Cref{Sec:queryAlgo}. Here, the algorithm must perform well for a \emph{large fraction} of inputs (under the given distribution), instead of being successful on all of them. The complexity of the optimal algorithm under a given distribution is called the average-case quantum query complexity.

\begin{definition}[Average-case query complexity]
  The \emph{average-case quantum query complexity}~$Q_{\dis}(\rel)$ of a relation~$\rel$ under a distribution~$\dis$ is the smallest integer $T$ such that there exists a quantum algorithm with query complexity $T$ that, given an input~$x$ sampled according to~$\dis$, outputs~$y \in \rel_x$ with probability at least~$2/3$. Similarly, we define $R_{\dis}(\rel)$ as the average-case randomized query complexity.
\end{definition}

\fakeparagraph{Query record.}
The core idea of the recording method is to construct a quantum object that keeps track of what queries an algorithm has made. If the queries were classical, it would suffice to take the record $R_t = ((i_1,x_{i_1}),(i_2,x_{i_2}),\dots,(i_t,x_{i_t}))$ of all the query-answer pairs seen by the algorithm after~$t$ queries (this is a random variable, which depends on the random input~$x$ and the random actions of the algorithm). This object is very helpful in quantifying the knowledge gained by an algorithm toward solving a problem. However, it becomes more difficult to define a similar record for quantum queries. For instance, a single quantum query can learn information about the whole input by querying all coordinates in superposition. The solution will be to construct a record that is itself a quantum state -- entangled with the state of the algorithm. The construction proceeds in three steps:

\begin{enumerate}
  \item \emph{Input purification:} the input distribution is encoded into a bipartite quantum state, one system being the state of the algorithm and the other being a purification register.
  \item \emph{Record state:} the purification register is mapped to a different basis in which it can be interpreted as a quantum record.
  \item \emph{Recording oracle:} the query operator is modified to operate directly on the record state with some desirable properties.
\end{enumerate}

We detail each step of the construction below. The reader can refer in parallel to the pictures in \Cref{Fig:purif,Fig:rec}, which summarize the main ideas of the process.

\fakeparagraph{1. Input purification.}
We first address the question of representing the state of an algorithm whose input is random. This will also be useful in \Cref{Sec:adver} when describing the adversary method. Recall that we denote the state of an algorithm after~$t$ queries on a fixed input~$x$ as~$\psx{t}$. There are two equivalent ways of representing the average state of the algorithm when~$x$ is chosen randomly with some probability $\dis(x)$,
  \[\rho^t = \sum_x \dis(x) \proj{\psi_x^t} \quad \longleftrightarrow \quad \pst{t} = \sum_x \sqrt{\dis(x)} \psx{t} \otimes \ket{x}.\]
The first representation uses the density matrix formalism. The system is described by the pure-state ensemble $\set{\dis(x),\psx{t}}_x$ corresponding to the density operator $\rho^t$. The second representation is a purification of $\rho^t$, where the purification register holds a copy of the input~$x$. In this section, we will work with the latter representation, as it retains the \emph{joint state} of the algorithm and input distribution (whereas~$\rho^t$ encodes only the marginal state on the algorithm registers).

We extend the formalism given in \Cref{Def:canon} to define a query algorithm using the joint state $\pst{t}$. The purification register acts as a control state for the joint oracle~$\ora^{\pm}$ defined as: $\ora^{\pm} \pt*{\ket{i,b} \otimes \ket{x}} = \pt*{\ora^{\pm}_x \ket{i,b}} \otimes \ket{x} = \omega^{b x_i} \ket{i,b} \otimes \ket{x}$. The unitaries $U_t$ applied by the algorithm cannot depend on $x$, hence they are extended to act as the identity on the purification register. The new formalism is given in the next definition and displayed in \Cref{Fig:purif}.

\begin{definition}[Joint state of a quantum query algorithm with input distribution]
  \label{Def:joint}
  A (memoryless) \emph{quantum query algorithm} with query complexity $T$ and input alphabet $\Sigma = \set{0,\dots,m-1}$ is a sequence $U_0,\dots,U_T$ of unitary operators acting on the Hilbert space spanned by the vectors~$\ket{i,b}$ where $i \in \ind$ and $b \in \Sigma$.

  The \emph{joint phase oracle} is $\ora^{\pm} = \sum_{x \in \Sigma^n} \ora^{\pm}_x \otimes \proj{x}$, and the \emph{joint state $\pst{t}$ of the algorithm and input} after $t \in \set{0,\dots,T}$ queries on an input distribution $\pt*{\dis(x)}_{x \in \Sigma^n}$ is defined as,
    \[\pst{t} = (U_t\otimes\id) \ora^{\pm} \dots (U_1\otimes\id) \ora^{\pm} (U_0\otimes\id) \pt[\Big]{ \ket{0,0} \otimes \sum_{x \in \Sigma^n} \sqrt{\dis(x)} \ket{x}}.\]
  Equivalently, $\pst{t} = \sum_{x \in \Sigma^n} \sqrt{\dis(x)} \psx{t}\otimes\ket{x}$ where $\psx{t} = U_t \ora^{\pm}_x \dots U_1 \ora^{\pm}_x U_0 \ket{0,0}$ is the state of the algorithm on the input $x$. The reduced density matrix representing the state of the algorithm after $t$ queries is denoted $\rho^t$ and is equal to $\sum_{x \in \Sigma^n} \dis(x) \proj{\psi_x^t}$.

  The \emph{output} of the algorithm is the random value $i$ obtained by measuring the index register of~$\pst{T}$ in the standard basis. The (average) \emph{success probability} $\psucc^{\dis}$ of the algorithm in computing a relation~$\rel$ on the input distribution~$\dis$ is the probability of measuring a correct output for a random input: $\psucc^{\dis} = \sum_{x \in \Sigma^n, i \in \rel_x}  \norm*{\pt*{\proj{i}\otimes \id \otimes \proj{x}}\pst{T}}^2$.
\end{definition}

\begin{figure}
  \centering
  \switchAMS{
  \begin{quantikz}[column sep=4.5mm]
    \lstick{Input register: $\sum_x \sqrt{\dis(x)} \ket{x}$}  &[-2mm] \qw \slice[style=black]{$\pst{0}$} & \gate[wires=3]{\ora^{\pm}} &[-2mm] \qw \slice[style=black]{$\pst{1}$} &  \ \ldots\ \qw & \gate[wires=3]{\ora^{\pm}} &[-2mm] \qw \slice[style=black]{$\pst{T}$} & \qw & [-5mm]  \\
    \lstick{Value register: \hspace{1.65cm} $\ket{0}$}  & \gate[wires=2]{U_0} & \qw & \gate[wires=2]{U_1} &  \ \ldots\ \qw & \qw & \gate[wires=2]{U_T} & \qw & \meter{} \\
    \lstick{Index register: \hspace{1.65cm} $\ket{0}$} & \qw & \qw & \qw & \ \ldots\ \qw & \qw & \qw & \qw & \qw
  \end{quantikz}
  }{
    \begin{quantikz}
      \lstick{Input register: $\sum_x \sqrt{\dis(x)} \ket{x}$}  &[-2mm] \qw \slice[style=black]{$\pst{0}$} & \gate[wires=3]{\ora^{\pm}} & \qw \slice[style=black]{$\pst{1}$} &  \ \ldots\ \qw & \gate[wires=3]{\ora^{\pm}} & \qw \slice[style=black]{$\pst{T}$} & \qw & [-5mm]  \\
      \lstick{Value register: \hspace{1.65cm} $\ket{0}$}  & \gate[wires=2]{U_0} & \qw & \gate[wires=2]{U_1} &  \ \ldots\ \qw & \qw & \gate[wires=2]{U_T} & \qw & \meter{} \rstick{\text{Output}} \\
      \lstick{Index register: \hspace{1.65cm} $\ket{0}$} & \qw & \qw & \qw & \ \ldots\ \qw & \qw & \qw & \qw & \qw
    \end{quantikz}
  }
  \caption{Canonical form of a (memoryless) quantum query algorithm with input distribution.}
  \label{Fig:purif}
\end{figure}
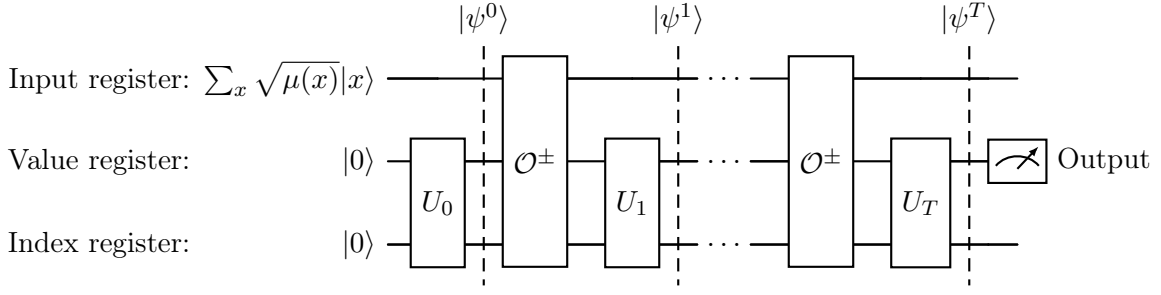

\fakeparagraph{2. Record state.}
We explain how to construct the quantum record by embedding the input register of $\pst{t}$ into a Hilbert space of higher dimension and applying a certain unitary on it. The input register is originally supported over the $m^n$ basis states $\ket{x} = \ket{x_1,\dots,x_n}$ where $x_i \in \set{0,\dots,m-1}$. The latter alphabet is augmented with a new ``empty'' symbol $\varnothing$, which will be used to represent the absence of knowledge from the algorithm about a coordinate of the input.
The \emph{record space} is the Hilbert space of dimension $(m+1)^n$ spanned by the vectors,
  \[\ket{x_1,\dots,x_n} = \ket{x_1}\otimes\cdots\otimes\ket{x_n} \quad \text{where}\quad x_i \in \set{0,\dots,m-1} \cup \set{\varnothing}.\]
Equivalently, it is the $n$-fold tensor product of the Hilbert space of dimension~$m+1$ spanned by the vectors $\ket{0}, \dots, \ket{m-1}, \ket{\varnothing}$. The input register of~$\pst{t}$ is renamed the \emph{record register} when it is interpreted as living in the record space. We aim to define a \emph{recording unitary}~$\tra$ -- operating on the record space -- that prepares a state $\psr{t} = (\id \otimes \tra) \pst{t}$ whose record register contains the \emph{record state} of the algorithm. This unitary will depend solely on the input distribution~$\unif$ and inherit its product structure. We can consider two situations that guide the choice of this operator:
\begin{itemize}
  \item \emph{(Initial state)} The initial state of the algorithm conveys no information on the input, so the record should start in the all-empty state $\ket{\varnothing}^{\otimes n}$. This corresponds to the initial joint state~$\pst{0}$ being mapped to,
     \[\pst{0} = (U_0 \ket{0,0}) \otimes \pt[\bigg]{\frac{1}{m^{n/2}} \sum_{x \in \Sigma^n} \ket{x}} \xmapsto{\id \otimes \tra} \psr{0} = (U_0 \ket{0,0}) \otimes \ket{\varnothing}^{\otimes n}.\]
  Hence, the unitary~$\tra$ shall map the uniform superposition to the state~$\ket{\varnothing}^{\otimes n}$.
  \item \emph{(Phase kickback)} If the algorithm queries the state $\ket{i,b} \otimes \pt[\big]{\frac{1}{m^{n/2}} \sum_{x \in \Sigma^n} \ket{x}}$ with~$b \neq 0$ then the post-query state is (up to the normalization factor~$1/m^{n/2}$):
  \switchAMS{
    \[\begin{array}{lrl}
      \ket{i,b} \otimes \sum_{x \in \Sigma^n} \ket{x}
       & \xmapsto{\ora^{\pm}} & \sum_{x \in \Sigma^n} \omega^{bx_i} \ket{i,b} \otimes \ket{x} \\[3mm]
       & = & \ket{i,b} \otimes \pt[\big]{\sum_{x_1\dots x_{i-1} \in \Sigma^{i-1}} \ket{x_1\dots x_{i-1}}} \\[3mm]
       & & \quad \otimes \pt[\big]{\sum_{x_i \in \Sigma} \omega^{bx_i} \ket{x_i}} \otimes \pt[\big]{\sum_{x_{i+1}\dots x_n \in \Sigma^{n-i}} \ket{x_{i+1}\dots x_n}}.
    \end{array}\]
  }{
    \[\begin{array}{lrl}
      \ket{i,b} \otimes \sum_{x \in \Sigma^n} \ket{x}
       & \xmapsto{\ora^{\pm}} & \sum_{x \in \Sigma^n} \omega^{bx_i} \ket{i,b} \otimes \ket{x} \\[3mm]
       & = & \ket{i,b} \otimes \pt[\big]{\sum_{x_1\dots x_{i-1} \in \Sigma^{i-1}} \ket{x_1\dots x_{i-1}}} \otimes \pt[\big]{\sum_{x_i \in \Sigma} \omega^{bx_i} \ket{x_i}} \\[3mm]
       & & \hspace{41mm} \otimes \pt[\big]{\sum_{x_{i+1}\dots x_n \in \Sigma^{n-i}} \ket{x_{i+1}\dots x_n}}.
    \end{array}\]
  }
  Hence, a query can equivalently be seen as a modification of the record state, rather than of the state of the algorithm. Furthermore, when querying the index~$i$, only the $i$-th subsystem of the record register is modified and becomes \emph{orthogonal} to the initial uniform superposition. This provides a natural criterion for including an input coordinate in the record state: if the state of its register is orthogonal to the uniform superposition, then the unitary~$\tra$ should keep it intact into the record; otherwise, it should replace it with an empty record,
    \[\sum_{x \in \Sigma^n} \omega^{bx_i} \ket{i,b} \otimes \ket{x} \xmapsto{\id \otimes \tra} \ket{i,b} \otimes \ket{\varnothing}^{\otimes i-1} \otimes \pt[\big]{\sum_{x_i \in \Sigma} \omega^{bx_i} \ket{x_i}} \otimes \ket{\varnothing}^{\otimes (n-i)}.\]
\end{itemize}

It is easy to find a unitary~$\tra$ that satisfies all of the above conditions. Due to the symmetries of the input distribution~$\unif$, it suffices to take the tensor product $\tra = \tra_1 \otimes \cdots \otimes \tra_n$ of~$n$ identical operators~$\tra_i$ acting on each subsystem of the record register as follows.

\begin{definition}[Recording operator]
  \label{Def:tra}
  Let $\tra_i$ denote the unitary and Hermitian operator acting on the $i$-th subsystem of the record space as follows:
    \[
    \tra_i :
    \left\{
        \begin{array}{lll}
            \frac{1}{\sqrt{m}} \sum_{x_i \in \Sigma} \ket{x_i} & \mapsto & \ket{\varnothing}, \\[3mm]
            \frac{1}{\sqrt{m}} \sum_{x_i \in \Sigma} \omega^{b x_i} \ket{x_i} & \mapsto & \frac{1}{\sqrt{m}} \sum_{x_i \in \Sigma} \omega^{b x_i} \ket{x_i} \quad \mbox{if $b \in \set{1,\dots,m-1}$}, \\[3mm]
            \ket{\varnothing} & \mapsto & \frac{1}{\sqrt{m}} \sum_{x_i \in \Sigma} \ket{x_i}.
        \end{array}
    \right.
    \]
  Given the joint state $\pst{t}$ of an algorithm and input distribution~$\unif$ (\Cref{Def:joint}), we define the \emph{joint state $\psr{t}$ of the algorithm and record} as,
    \[\psr{t} = \pt*{\id \otimes (\tra_1 \otimes \cdots \otimes \tra_n)} \pst{t} = \pt*{\id \otimes \tra} \pst{t}\]
  where $\tra = \tra_1 \otimes \cdots \otimes \tra_n$ is the \emph{recording operator} acting on the record space.
\end{definition}

As a first simple observation, the size of the quantum record (i.e., the maximum number of non-$\varnothing$ entries in the basis states over which it is supported) cannot be larger than the number of quantum queries made so far.

\begin{fact}[Record size]
  \label{Fact:rec}
  The state $\psr{t} = \pt*{\id \otimes \tra} \pst{t}$ after~$t$ queries is supported only over basis states $\ket{i,b}\otimes\ket{x_1,\dots,x_n}$ whose record size is at most
  $\abs{\set{j : x_j \neq \varnothing}} \leq t$.
\end{fact}

\begin{proof}
  By tracking the action of the phase oracle~$\ora^{\pm}$ during~$t$ queries, one can see that the joint state of the algorithm and input admits a decomposition of the form,
    \[\pst{t} = \sum_{i \in \set{1,\dots,n},b \in \Sigma,c \in \Sigma^n} \alpha_{i,b,c} \ket{i,b} \otimes \sum_{x \in \Sigma^n} \omega^{c_1 x_1 + \dots + c_n x_n} \ket{x}\]
  for some complex coefficients~$\alpha_{i,b,c}$ that can be non-zero only when~$\abs{\set{j : c_j \neq 0}} \leq t$. For a given~$c$, the state $\sum_{x \in \Sigma^n} \omega^{c_1 x_1 + \dots + c_n x_n} \ket{x} = \pt{\sum_{x_1 \in \Sigma} \omega^{c_1 x_1} \ket{x_1}} \otimes \cdots \otimes \pt{\sum_{x_n \in \Sigma} \omega^{c_n x_n} \ket{x_n}}$ is mapped by the recording operator~$\tra$ to the product state~$\ket{R_1}\otimes\dots\otimes\ket{R_n}$ where $\ket{R_j} = \sqrt{m}\ket{\varnothing}$ if $c_j = 0$, and $\ket{R_j} = \sum_{x_j \in \Sigma} \omega^{c_j x_j} \ket{x_j}$ otherwise. Hence, $\psr{t} = \pt*{\id \otimes \tra} \pst{t}$ is supported over basis states that have at most~$t$ coordinates different from~$\varnothing$.
\end{proof}

The quantum record can behave very differently from its classical counterpart. For instance, its size can decrease over time if the algorithm uses query parameters~$(i,b)$ that make the $i$-th subsystem return to the uniform superposition. This phenomenon is unavoidable to preserve the reversibility of the computation. The purpose of the next section is to better understand how the record evolves after each query.

\fakeparagraph{3. Recording oracle.}
It is not clear yet that the quantum record has any meaningful properties beyond that of \Cref{Fact:rec}. We provide a different viewpoint on the construction of the quantum record, which better explains its evolution during the computation. To that end, we define a third query operator -- the \emph{recording oracle}~$\rec$ -- that conjugates the phase oracle~$\ora^{\pm}$ with the recording operator~$\id \otimes \tra$. Unlike the binary and phase oracles, the recording oracle depends on the chosen input distribution encoded into~$\tra$.

\begin{definition}[Recording oracle]
  The \emph{recording oracle} for the input distribution~$\unif$ is the unitary operator acting on the joint state of the algorithm and record as follows,
    \[\rec = \pt*{\id \otimes \tra} \ora^{\pm} \pt*{\id \otimes \tra},\]
  where the phase oracle is extended to the record space as $\ora^{\pm} = \sum_{x \in \pt{\Sigma \cup \set{\varnothing}}^n} \ora^{\pm}_x \otimes \proj{x}$ with $\ora^{\pm}_x \ket{i,b} = \omega^{bx_i} \ket{i,b}$ when $x_i \in \Sigma$, and $\ora^{\pm}_x \ket{i,b} = \ket{i,b}$ when~$x_i = \varnothing$.
\end{definition}

The need to specify the behavior of the phase oracle on~$x_i = \varnothing$ is an artifact of the construction. This case will never occur in practice, but it makes the analysis simpler.

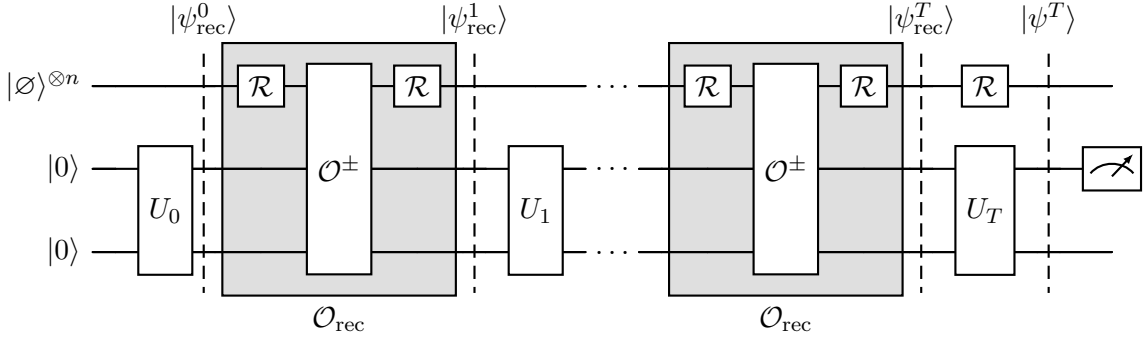
\begin{figure}
  \centering
  \switchAMS{    
  \begin{quantikz}[column sep = 2.5mm]
    \lstick{$\ket{\varnothing}^{\otimes n}$} &[-3mm] \qw
    & \qw \slice[style=black, label style={yshift=1mm}]{$\psr{0}$}
    & \qw & \gate{\tra} \gategroup[3,steps=3,style={fill=gray!25,inner xsep=2pt},background,label style={label position=below,anchor=north,yshift=-0.2cm}]{{$\rec$}}
    & \gate[wires=3]{\ora^{\pm}} & \gate{\tra}
    & \qw \slice[style=black, label style={yshift=1mm}]{$\psr{1}$}
    & \qw &[-2mm] \qw & \ \ldots\ \qw & \qw
    & \gate{\tra} \gategroup[3,steps=3,style={fill=gray!25,inner xsep=2pt},background,label style={label position=below,anchor=north,yshift=-0.2cm}]{{$\rec$}}
    & \gate[wires=3]{\ora^{\pm}} & \gate{\tra}
    & \qw \slice[style=black, label style={yshift=1mm}]{$\psr{T}$}
    & \qw &[-2mm] \gate{\tra} & \qw \slice[style=black, label style={yshift=1mm}]{$\pst{T}$}
    &[-2mm] \qw &[-1mm] \qw \\
    \lstick{$\ket{0}$} & \qw & \gate[wires=2]{U_0} & \qw & \qw & \qw & \qw & \qw & \qw & \gate[wires=2]{U_1} & \ \ldots\ \qw & \qw & \qw & \qw & \qw & \qw & \qw & \gate[wires=2]{U_T} &  \qw & \qw & \meter{}  \\
    \lstick{$\ket{0}$} & \qw & \qw & \qw & \qw &  \qw & \qw & \qw & \qw & \qw & \ \ldots\ \qw & \qw & \qw & \qw & \qw & \qw & \qw & \qw & \qw & \qw & \qw
  \end{quantikz}
  }{
    \begin{quantikz}[column sep = 3mm]
      \lstick{$\ket{\varnothing}^{\otimes n}$} & \qw
      & \qw \slice[style=black, label style={yshift=1mm}]{$\psr{0}$}
      & \qw & \gate{\tra} \gategroup[3,steps=3,style={fill=gray!25,inner xsep=2pt},background,label style={label position=below,anchor=north,yshift=-0.2cm}]{{$\rec$}}
      & \gate[wires=3]{\ora^{\pm}} & \gate{\tra}
      & \qw \slice[style=black, label style={yshift=1mm}]{$\psr{1}$}
      & \qw & \qw & \ \ldots\ \qw & \qw
      & \gate{\tra} \gategroup[3,steps=3,style={fill=gray!25,inner xsep=2pt},background,label style={label position=below,anchor=north,yshift=-0.2cm}]{{$\rec$}}
      & \gate[wires=3]{\ora^{\pm}} & \gate{\tra}
      & \qw \slice[style=black, label style={yshift=1mm}]{$\psr{T}$}
      & \qw & \gate{\tra} & \qw \slice[style=black, label style={yshift=1mm}]{$\pst{T}$}
      & \qw & \qw \\
      \lstick{$\ket{0}$} & \qw & \gate[wires=2]{U_0} & \qw & \qw & \qw & \qw & \qw & \qw & \gate[wires=2]{U_1} & \ \ldots\ \qw & \qw & \qw & \qw & \qw & \qw & \qw & \gate[wires=2]{U_T} &  \qw & \qw & \meter{}  \\
      \lstick{$\ket{0}$} & \qw & \qw & \qw & \qw &  \qw & \qw & \qw & \qw & \qw & \ \ldots\ \qw & \qw & \qw & \qw & \qw & \qw & \qw & \qw & \qw & \qw & \qw
    \end{quantikz}
  }
  \caption{Canonical form of a (memoryless) quantum query algorithm with recording oracle.}
  \label{Fig:rec}
\end{figure}

The recording oracle~$\rec$ has two central properties that are at the core of the recording method. First, the joint state~$\psr{t}$ can equally be viewed as the state obtained by replacing the phase oracle with the recording oracle in the algorithm. This is represented in \Cref{Fig:rec}. Second, the evolution of the record state upon querying~$\rec$ follows \emph{almost} the same rules as a classical record: if a coordinate is not in the record prior to the query then a uniform superposition is recorded; and if it is already in the record, the state remains (almost) unchanged.

\begin{theorem}[Recording method]
  \label{Thm:recording}
  Consider a quantum algorithm that accesses a random input drawn from the uniform distribution~$\unif$. Let~$\pst{t}$ denote the joint state of the algorithm and input after~$t$ queries (\Cref{Def:joint}).
  Then,
  \begin{itemize}
    \item (Indistinguishability) The joint state $\psr{t} = (\id \otimes \tra) \pst{t}$ of the algorithm and record is also equal to
    $\psr{t} = (U_t \otimes \id) \rec (U_{t-1} \otimes \id) \rec (U_0 \otimes \id) \pt{\ket{0,0}\otimes\ket{\varnothing,\dots,\varnothing}}.$
    \item (Recording action) The action of the recording oracle~$\rec$ on a basis state~$\ket{i,b} \otimes \ket{x_1,\dots,x_n}$ where $i \in \ind$, $b \in \Sigma \setminus \set{0}$ and $x \in \pt*{\Sigma \cup \set{\varnothing}}^m$ is given by the equations,
  \end{itemize}
  \switchAMS{
    \begin{align*}
    & \rec \ket{i,b} \otimes \ket{\dots x_{i-1},x_i,x_{i+1} \dots} = \\[2mm]
    & \qquad \left\{
      \begin{array}{l@{\hspace{5mm}}l}
        \ket{i,b} \otimes \ket{\dots x_{i-1}}\pt[\Big]{\frac{1}{\sqrt{m}} \sum_{x_i' \in \Sigma} \omega^{bx_i'}\ket{x_i'}}\ket{x_{i+1} \dots} & \text{if $x_i = \varnothing$} \\[4mm]
        \ket{i,b} \otimes \ket{\dots x_{i-1}}\pt[\big]{\omega^{bx_i}\ket{x_i} + \ket{\mathrm{error}_{x_i}}}\ket{x_{i+1}\dots} & \text{if $x_i \in \Sigma$}
      \end{array}
  \right. \\[2mm]
  & \text{where $\ket{\mathrm{error}_{x_i}} = \frac{\omega^{bx_i}}{\sqrt{m}}\ket{\varnothing} + \sum\limits_{x_i' \in \Sigma} \frac{1-\omega^{bx_i}-\omega^{bx_i'}}{m} \ket{x_i'}$.}
\end{align*}
  }{
    \[(x_i = \varnothing)\quad \rec \ket{i,b} \otimes \ket{\dots x_{i-1},\varnothing,x_{i+1}\dots}
    = \ket{i,b} \otimes \ket{\dots x_{i-1}}\pt[\Big]{\frac{1}{\sqrt{m}} \sum_{x_i' \in \Sigma} \omega^{bx_i'}\ket{x_i'}}\ket{x_{i+1} \dots},\]
    \[(x_i \in \Sigma)\quad \rec \ket{i,b} \otimes \ket{\dots x_{i-1},x_i,x_{i+1} \dots}
    = \ket{i,b} \otimes \ket{\dots x_{i-1}}\pt[\big]{\omega^{bx_i}\ket{x_i} + \ket{\mathrm{error}_{x_i}}}\ket{x_{i+1}\dots},\]
    \null\hfill where $\ket{\mathrm{error}_{x_i}} = \frac{\omega^{bx_i}}{\sqrt{m}}\ket{\varnothing} + \sum\limits_{x_i' \in \Sigma} \frac{1-\omega^{bx_i}-\omega^{bx_i'}}{m} \ket{x_i'}$.
  }
  
  \hspace*{3mm} If~$b = 0$ then $\rec$ makes no change to the state~$\ket{i,b} \otimes \ket{x_1,\dots,x_n}$.
\end{theorem}

\begin{proof}[Proof of the indistinguishability]
  The proof is by induction on~$t$. The base case ($t=0$) uses that~$U_0$ and~$\tra$ are acting on different registers, hence their order can be inverted: $\psr{0} = (\id \otimes \tra) \pst{0} = (\id \otimes \tra) (U_0 \otimes \id) \pt[\big]{\ket{0,0} \otimes \frac{1}{m^{n/2}} \sum_{x \in \Sigma^n} \ket{x}} = (U_0 \otimes \id) (\id \otimes \tra)  \pt[\big]{\ket{0,0} \otimes \frac{1}{m^{n/2}} \sum_{x \in \Sigma^n} \ket{x}} =  (U_0 \otimes \id)  (\ket{0,0}\otimes \allowbreak \ket{\varnothing,\dots,\varnothing})$. The induction step applies the same argument to $U_{t+1}$ and uses the fact that~$\tra$ squares to the identity: $\psr{t+1} = (\id \otimes \tra) (U_{t+1} \otimes \id) \ora^{\pm} \pst{t} = (U_{t+1} \otimes \id) (\id \otimes \tra) \ora^{\pm} \pst{t} = (U_{t+1} \otimes \id) (\id \otimes \tra) \ora^{\pm} (\id \otimes \tra) (\id \otimes \tra)\pst{t} = (U_{t+1} \otimes \id) \rec \psr{t}$.
\end{proof}

\begin{proof}[Proof of the recording action]
  We assume~$b \neq 0$, as it is easy to see that $\rec$ acts as the identity otherwise. We decompose the action of $\rec = \pt*{\id \otimes \tra} \ora^{\pm} \pt*{\id \otimes \tra}$ when $x_i = \varnothing$,
  \switchAMS{
    \begin{align*}
      & \ket{i,b} \otimes \ket{\dots x_{i-1}, \varnothing,x_{i+1} \dots} \\
      & \hspace{2cm} \xmapsto{\id \otimes \tra} \ket{i,b} \otimes (\dots\otimes \tra_{i-1}\ket{x_{i-1}})\pt[\Big]{\frac{1}{\sqrt{m}} \sum_{x_i' \in \Sigma}\ket{x_i'}}(\tra_{i+1} \ket{x_{i+1}} \otimes \dots) \\
      & \hspace{2cm} \xmapsto{\hspace{-0.2mm}\ora^{\pm}\hspace{-1.7mm}} \ket{i,b} \otimes (\dots \otimes \tra_{i-1}\ket{x_{i-1}})\pt[\Big]{\frac{1}{\sqrt{m}} \sum_{x_i' \in \Sigma} \omega^{bx_i'}\ket{x_i'}}(\tra_{i+1} \ket{x_{i+1}} \otimes \dots) \\
      & \hspace{2cm} \xmapsto{\id \otimes \tra} \ket{i,b} \otimes \ket{\dots x_{i-1}}\pt[\Big]{\frac{1}{\sqrt{m}} \sum_{x_i' \in \Sigma} \omega^{bx_i'}\ket{x_i'}}\ket{x_{i+1} \dots}.
    \end{align*}
  }{
    \begin{align*}
      \ket{i,b} \otimes \ket{\dots x_{i-1}, \varnothing,x_{i+1} \dots}
      & \xmapsto{\id \otimes \tra} \ket{i,b} \otimes (\dots\otimes \tra_{i-1}\ket{x_{i-1}})\pt[\Big]{\frac{1}{\sqrt{m}} \sum_{x_i' \in \Sigma}\ket{x_i'}}(\tra_{i+1} \ket{x_{i+1}} \otimes \dots) \\
      & \xmapsto{\hspace{-0.2mm}\ora^{\pm}\hspace{-1.7mm}} \ket{i,b} \otimes (\dots \otimes \tra_{i-1}\ket{x_{i-1}})\pt[\Big]{\frac{1}{\sqrt{m}} \sum_{x_i' \in \Sigma} \omega^{bx_i'}\ket{x_i'}}(\tra_{i+1} \ket{x_{i+1}} \otimes \dots) \\
      & \xmapsto{\id \otimes \tra} \ket{i,b} \otimes \ket{\dots x_{i-1}}\pt[\Big]{\frac{1}{\sqrt{m}} \sum_{x_i' \in \Sigma} \omega^{bx_i'}\ket{x_i'}}\ket{x_{i+1} \dots}.
    \end{align*}
  }
  The last step uses that $\tra_j^2 = \id$ for all $j$, hence all registers are restored to their initial basis state, except the $i$-th record register.

  We now consider the case of $x_i \in \Sigma$. For the ease of notation, we only track the evolution of the $i$-th record register (the other registers do not change, for the same reasons as above),
    \begin{align*}
      \ket{x_i}
      & \xmapsto{\tra_i} \ket{x_i} + \frac{1}{\sqrt{m}} \ket{\varnothing} - \frac{1}{n}\sum_{x_i' \in \Sigma} \ket{x_i'} \\
      & \xmapsto{\ora^{\pm}} \omega^{b x_i} \ket{x_i} + \frac{1}{\sqrt{m}} \ket{\varnothing} - \frac{1}{n}\sum_{x_i' \in \Sigma} \omega^{b x_i'}\ket{x_i'} \\
      & \xmapsto{\tra_i} \omega^{bx_i}\ket{x_i} + \frac{\omega^{bx_i}}{\sqrt{m}} \ket{\varnothing} - \frac{1}{m}\sum_{x_i' \in \Sigma} \omega^{bx_i}\ket{x_i'} + \frac{1}{m} \sum_{x_i' \in \Sigma} \ket{x_i'} - \frac{1}{m}\sum_{x_i' \in \Sigma} \omega^{bx_i'}\ket{x_i'}.
    \end{align*}
  The first step has been obtained by rewriting the action of the recording operator (\Cref{Def:tra}) in the standard basis.
\end{proof}


  \subsection{Applications}

We illustrate the recording method on two problems that are ubiquitous in cryptography: finding a preimage or a collision in a random (hash) function. Beyond establishing the query complexity of these problems, the recording method will give very tight upper bounds on the best average success probability that can be achieved with a given number of quantum queries. This refinement is useful, for instance, when choosing the security parameters (e.g., key length) of cryptographic schemes based on hash functions.


\fakeparagraph{Application 1: The \srcf\ problem.}
This problem is a variant of \orf, where instead of deciding if the input contains a particular value (e.g., the bit~$1$), the task is to find such a value (when it exists). We will be studying the hardness of this problem under the uniform input distribution.

\begin{definition}[\srcf]
  The $\srcf$ problem is to find an index $i \in \ind$ such that~$x_i = 1$ in an input~$x \in \set{0,\dots,m-1}^n$.
\end{definition}

The lower bound will be given as a function of the alphabet size~$m$. The input size~$n$ appears indirectly in the result, as it constrains the values of~$m$ for which the problem is non-trivial under the uniform distribution (for instance, if $m \gg n$ then a random input has no solution at all with high probability). If it helps the reader, one can fix~$m = n/2$ since it guarantees the existence of two solutions in expectation, and the probability of having no solution is small $(1-2/n)^n \leq e^{-2} \leq 1/7$ (this event can be ignored as long as its probability is below the allowed failure probability, e.g.~$1/3$).

It is straightforward to argue that classical algorithms need~$\om{m}$ queries to solve the $\srcf$ problem. We provide a detailed proof that will help understanding the quantum case afterward.

\begin{proposition}[Classical recording method applied to \srcf]
  The average-case randomized complexity of the \srcf\ problem under the uniform distribution is at least~$R_{\unif}(\srcf) \geq 2m/3-1$.
\end{proposition}

\begin{proof}
  We measure the progress of an algorithm by the probability of the event $E_t$: ``the algorithm queries at least one coordinate equals to $x_i = 1$ during its first $t$ queries''. The probability that the $t+1$-th query triggers the event is at most $\pr{E_{t+1} \given \br{E_t}} \leq 1/m$, irrespective of the behavior of the algorithm (since the $x_i$'s are independent, identically distributed random variables). Hence, the progress behaves as follows,
    \begin{itemize}
      \item (Initial condition) $\pr{E_0} = 0$,
      \item (Progress evolution) $\pr{E_{t+1}} = \pr{E_{t}} + \pr{E_{t+1} \given \br{E_t}} \cdot \pr{\br{E_t}} \leq \pr{E_t} + 1/m$.
    \end{itemize}
  It remains to relate the average success probability $\psucc^{\unif}$ of an algorithm making~$T$ queries to its final progress $\pr{E_T}$. If the algorithm never queried a coordinate equal to $1$ (i.e., the event~$E_T$ does not happen), then it is left to guess randomly where such a coordinate can be. This is somewhat the same as the event $E_{T+1}$ given $\br{E_T}$, except that the algorithm cannot see the result of the last query.
    \begin{itemize}
      \item (Final condition) $\psucc^{\unif} \leq \pr{E_T} + 1/m$.
    \end{itemize}
  By combining the three bullet points together, we obtain that the best possible success probability after $T$ queries must be at most $\psucc^{\unif} \leq (T+1)/m$. In particular, succeeding with average probability at least~$\psucc^{\unif} \geq 2/3$ requires making at least~$T \geq 2m/3-1$ queries.
\end{proof}

We are now going to mimic the above proof in the quantum query model, using the recording method formalism. The main challenge is to adapt the progress measure when the query record (hence, the event $E_t$) is no longer properly defined as a random variable. The solution is to measure the progress as the \emph{total amplitude} (norm) of the part of the quantum record that contains a solution. The quadratic decrease in the lower bound can be traced back to manipulating amplitudes in the proof, rather than probabilities (squared norm).

\begin{proposition}[Quantum recording method applied to \srcf]
  \label{Prop:search}
  The average-case quantum complexity of the \srcf\ problem under the uniform distribution is at least~$Q_{\unif}(\srcf) \geq \sqrt{m/15}-\sqrt{1/5}$.
\end{proposition}

\begin{proof}
  Consider any quantum algorithm that solves the \srcf\ problem using some number~$T$ of queries. Define the progress measure $\Delta_t$ after $t \in \set{0,\dots,T}$ queries as the norm of the state obtained by projecting~$\psr{t}$ (\Cref{Def:tra}) onto the records containing at least one coordinate equals to~$x_i = 1$,
    \[\Delta_t = \norm{\Pi \psr{t}} \quad \text{where} \quad \Pi = \id \otimes \sum_{\substack{x \in \pt{\Sigma \cup \set{\varnothing}}^n, \\\exists i, x_i = 1}} \proj{x}.\]
  We show that the progress $\Delta_t$ and success probability $\psucc^{\unif}$ of the algorithm obey the following inequalities,
  \begin{itemize}
    \item (Initial condition) $\Delta_0 = 0$,
    \item (Progress evolution) $\Delta_{t+1} \leq \Delta_t + \sqrt{10/m}$,
    \item (Final condition) $\psucc^{\unif} \leq \pt{\Delta_T + \sqrt{2/m}}^2$.
  \end{itemize}
  It follows immediately that $\psucc^{\unif} \leq (\sqrt{10}T+\sqrt{2})^2/m$. Hence, succeeding with probability at least $\psucc^{\unif} \geq 2/3$ requires making at least~$T \geq \sqrt{m/15}-\sqrt{1/5}$ quantum queries.

  \textit{Proof of the initial condition.} The record is initially empty, $\psr{0} = (U_0 \ket{0,0}) \otimes \ket{\varnothing}^{\otimes n}$, hence the progress starts at $\Delta_0 = \norm{\Pi \psr{0}} = 0$.

  \textit{Proof of the progress evolution.} We first prove an analogous statement to the equality $\pr{E_{t+1}} = \pr{E_{t}} + \pr{E_{t+1} , \br{E_t}}$ used in the classical setting. We substitute the use of the law of total probability with the triangle inequality. The progress increases at most by the norm of the part of the state recording a value~$x_i = 1$ for the first time.
  \switchAMS{
    \begin{flalign*}
      && \Delta_{t+1}
        & = \norm{\Pi (U_{t+1} \otimes \id) \rec \psr{t}} & \eqname[4cm]{definition of $\psr{t+1}$} \\
      &&  & = \norm{(U_{t+1} \otimes \id) \Pi \rec \psr{t}} &  \eqname{$\Pi$ and $U_{t+1} \otimes \id$ commute}\\
      &&  & = \norm{\Pi \rec \psr{t}} & \eqname{unitary invariance of the norm}\\
      &&  & \leq \norm{\Pi \rec \Pi \psr{t}} + \norm{\Pi \rec (\id - \Pi) \psr{t}} & \eqname{triangle inequality} \\
      &&  & \leq \norm{\Pi \psr{t}} + \norm{\Pi \rec (\id - \Pi) \psr{t}} & \eqname{submultiplicativity of the norm} \\
      &&  & = \Delta_t + \norm{\Pi \rec (\id - \Pi) \psr{t}}.
    \end{flalign*}
  }{ 
    \begin{align*}
      \Delta_{t+1}
        & = \norm{\Pi (U_{t+1} \otimes \id) \rec \psr{t}} \tag{definition of $\psr{t+1}$} \\
        & = \norm{(U_{t+1} \otimes \id) \Pi \rec \psr{t}} \tag{$\Pi$ and $U_{t+1} \otimes \id$ commute}\\
        & = \norm{\Pi \rec \psr{t}} \tag{unitary invariance of the norm}\\
        & \leq \norm{\Pi \rec \Pi \psr{t}} + \norm{\Pi \rec (\id - \Pi) \psr{t}} \tag{triangle inequality} \\
        & \leq \norm{\Pi \psr{t}} + \norm{\Pi \rec (\id - \Pi) \psr{t}} \tag{submultiplicativity of the norm} \\
        & = \Delta_t + \norm{\Pi \rec (\id - \Pi) \psr{t}}.
    \end{align*}
  }
  Next, we show that $\norm{\Pi \rec (\id - \Pi) \psr{t}} \leq \sqrt{10/m} \norm{(\id - \Pi) \psr{t}}$, in analogy to the statement $\pr{E_{t+1} , \br{E_t}} = \pr{E_{t+1} \given \br{E_t}} \cdot \pr{\br{E_t}} \leq 1/m \pr{\br{E_t}}$. The proof is carried out in greater generality, replacing $(\id - \Pi) \psr{t}$ with any state $\ket{\psi} = \sum_{i,b,x} \alpha_{i,b,x} \ket{i,b} \otimes \ket{x}$ in the support of $\id - \Pi$ (i.e., no record in the support of~$\ket{\psi}$ shall contain the value $1$). We decompose such a state into $n+2$ mutually orthogonal components $\ket{\psi} = \ps{\mathrm{id}} + \ps{\varnothing} + \sum_{y \in \Sigma} \ps{y}$ defined as follows:
    \begin{itemize}
      \item[--] $\ps{\mathrm{id}} = \sum_{i,b,x : b = 0} \alpha_{i,b,x} \ket{i,b} \otimes \ket{x}$ \hfill (null query),
      \item[--] $\ps{\varnothing} = \sum_{i,b,x : x_i = \varnothing, b \neq 0} \alpha_{i,b,x} \ket{i,b}\otimes \ket{x}$ \hfill (non-null query, empty record),
      \item[--] $\ps{y} = \sum_{i,b,x : x_i = y, b \neq 0} \alpha_{i,b,x} \ket{i,b}\otimes \ket{x}$ \hfill (non-null query, nonempty record).
    \end{itemize}
  The component~$\ps{1}$ is zero by definition of~$\id - \Pi$. We analyze how much the norms of the other components decrease after applying~$\Pi \rec$. The action of the oracle~$\rec$ on a basis state is dictated by \Cref{Thm:recording}. Notice that the only way for these states to be in the support of~$\Pi$ after applying~$\rec$ is to record $x_i = 1$ at the position~$i$ indicated by the index register, since the rest of the record stays unchanged. The norms are:
    \begin{itemize}
      \item[--] $\norm{\Pi \rec \ps{\mathrm{id}}} = 0$: The state becomes zero since the record does not change when the value register holds a zero.
      \item[--] $\norm{\Pi \rec \ps{\varnothing}} = \frac{1}{\sqrt{m}} \norm{\ps{\varnothing}}$: The  state is equal to
        $\Pi \rec \ps{\varnothing} = \sum_{x_i = \varnothing, b \neq 0} \alpha_{i,b,x} \allowbreak \frac{\omega^{b}}{\sqrt{m}} \ket{i,b} \otimes \ket{x^{\set{i}}}$
      with $x^{\set{i}}_i = 1$ and $x^{\set{i}}_j = x_j$ if $j \neq i$. Thus, $\norm{\Pi \rec \ps{\varnothing}}^2 = \sum_{x_i = \varnothing, b \neq 0} \frac{\abs{\alpha_{i,b,x}}^2}{m} = \frac{1}{m} \norm{\ps{\varnothing}}^2$.
      \item[--] $\norm{\Pi \rec \ps{y}} \leq \frac{3}{m} \norm{\ps{y}}$: The state is equal to
        $\Pi \rec \ps{y} = \sum_{ x_i = y, b \neq 0} \alpha_{i,b,x} \allowbreak \frac{1-\omega^{b}-\omega^{by}}{m} \ket{i,b} \otimes \ket{x^{\set{i}}}$.
      Thus, $\norm{\Pi \rec \ps{y}}^2 = \sum_{ x_i = y, b \neq 0} \frac{|(1-\omega^{b}-\omega^{by})\alpha_{i,b,x}|^2}{m^2} \leq \frac{9}{m^2} \norm{\ps{y}}^2$.
    \end{itemize}
   Finally, using the triangle and Cauchy--Schwarz inequalities, we conclude that the overall state has norm $\norm{\Pi \rec \ket{\psi}} \leq \norm{\Pi \rec \ps{\varnothing}} + \sum_{y \in \Sigma \setminus \set{1}} \norm{\Pi \rec \ps{y}} \leq \frac{1}{\sqrt{m}} \norm{ \ps{\varnothing}} + \frac{3}{m} \sum_{y \in \Sigma \setminus \set{1}} \norm{ \ps{y}} \leq \sqrt{\frac{10}{m}} \norm{\ket{\psi}}$.

  \textit{Proof of the final condition.} The average success probability is defined as $\psucc^{\unif} = \norm{\prsucc\pst{T}}^2$ where $\prsucc$ is the projector onto the states $\ket{i,b} \otimes \ket{x}$ with $x_i = 1$. Using the relationship between the joint states $\pst{T}$ and $\psr{T}$ (\Cref{Def:tra}), we have
  \switchAMS{
  \begin{flalign*}
    && \psucc^{\unif}
      & = \norm{\prsucc(\id \otimes \tra) \psr{T}}^2 & \eqname[12mm]{since $\pst{T} = (\id \otimes \tra) \psr{T}$}\\
    &&  & \leq \pt*{\norm{\prsucc(\id \otimes \tra) \Pi \psr{T}} + \norm{\prsucc(\id \otimes \tra) (\id - \Pi) \psr{T}}}^2 \\ \intertext{\raggedleft (triangle inequality)}
    &&  & \leq \pt*{\norm{\Pi \psr{T}} + \norm{\prsucc(\id \otimes \tra) (\id - \Pi) \psr{T}}}^2 \\ \intertext{\raggedleft (submultiplicativity of the norm)}
    && & = \pt*{\Delta_T + \norm{\Pi (\id \otimes \tra) (\id - \Pi) \psr{T}}}^2.
  \end{flalign*}
  }{
    \begin{align*}
      \psucc^{\unif}
        & = \norm{\prsucc(\id \otimes \tra) \psr{T}}^2 \tag{since $\pst{T} = (\id \otimes \tra) \psr{T}$}\\
        & \leq \pt*{\norm{\prsucc(\id \otimes \tra) \Pi \psr{T}} + \norm{\prsucc(\id \otimes \tra) (\id - \Pi) \psr{T}}}^2 \tag{triangle inequality}\\
        & \leq \pt*{\norm{\Pi \psr{T}} + \norm{\prsucc(\id \otimes \tra) (\id - \Pi) \psr{T}}}^2 \tag{submultiplicativity of the norm} \\
        & = \pt*{\Delta_T + \norm{\Pi (\id \otimes \tra) (\id - \Pi) \psr{T}}}^2.
    \end{align*}
  }
   The analysis of $\norm{\Pi (\id \otimes \tra) (\id - \Pi) \psr{T}}$ follows that of $\norm{\Pi \rec (\id - \Pi) \psr{T}}$ done previously. Using the same notation as before $\ket{\psi} = \ps{\mathrm{id}} + \ps{\varnothing} + \sum_{y \in \Sigma \setminus \set{1}} \ps{y}$ for a state in the support of $\id-\Pi$, we have $\norm{\Pi (\id \otimes \tra) \ps{\mathrm{id}}} = 0$, $\norm{\Pi (\id \otimes \tra) \ps{\varnothing}} = \frac{1}{\sqrt{m}} \norm{\ps{\varnothing}}$ and $\norm{\Pi (\id \otimes \tra) \ps{y}} = \frac{1}{m} \norm{\ps{y}}$ (the action of $\tra$ in the standard basis is given, for instance, in the proof of \Cref{Thm:recording}). We conclude that $\norm{\Pi (\id \otimes \tra) \ket{\psi}} \leq \sqrt{\frac{2}{m}} \norm{\ket{\psi}}$.
\end{proof}


\fakeparagraph{Application 2: The \colf\ problem.}
We sketch a second application of the recording method for the \colf\ problem, defined as follows.

\begin{definition}[\colf]
  The $\colf$ problem is to find two distinct indices $i \neq j \in \ind$ such that~$x_i = x_j$ in an input~$x \in \set{0,\dots,m-1}^n$.
\end{definition}

Again, there is a range of parameters $n,m$ for which the problem is relevant under the uniform distribution. For instance, if~$m = n$ then there are $\frac{1}{m}\binom{n}{2} = (n-1)/2$ collision pairs in expectation, one of which can be found using~$\bo{\sqrt{m}}$ classical queries (birthday attack) or~$\bo{m^{1/3}}$ quantum queries (BHT~\cite{BHT98c} or Ambainis~\cite{Amb07j} algorithms).

The classical complexity of the \colf\ problem can be understood using the event $E_t$: ``the algorithm queries two different coordinates with equal values $x_i = x_j$ during its first $t$ queries''. The progress evolves as $\pr{E_{t+1}} = \pr{E_{t}} + t/m$ since the probability that the $t+1$-th query returns one of the~$t$ values observed before (i.e., produces a collision) is at most $t/m$. This yields that the progress after~$T$ queries is $\Delta_T = 1/m + 2/m + \dots + (T-1)/m = \bo{T^2/m}$, hence requiring $T = \om{\sqrt{m}}$ queries to make it sufficiently large. We can transpose this proof to the quantum setting using the following progress measure.

\begin{proposition}[Quantum recording progress for \colf]
  Consider any quantum algorithm with access to a random input drawn from the uniform distribution~$\unif$. Let~$\psr{t}$ denote the joint state of the algorithm and record after~$t$ queries to the input. Define the progress measure~$\Delta_t$ as the norm of the state obtained by projecting~$\psr{t}$ onto the records containing at least two equal coordinates~$x_i = x_j$,
    \[\Delta_t = \norm{\Pi \psr{t}} \quad \text{where} \quad \Pi = \id \otimes \sum_{\substack{x \in \pt{\Sigma \cup \set{\varnothing}}^n, \\\exists i \neq j, x_i = x_j \neq \varnothing}} \proj{x}.\]
  Then the progress $\Delta_t$ obeys the inequality $\Delta_{t+1} \leq \Delta_t + \sqrt{\frac{10t}{m}}$.
\end{proposition}

\begin{proof}
  Using the same argument as in the proof of \Cref{Prop:search}, the progress increases after each query by at most $\Delta_{t+1} \leq \Delta_t + \norm{\Pi \rec (\id - \Pi) \psr{t}}$. Let $\ket{\psi} = \sum_{i,b,x} \alpha_{i,b,x} \ket{i,b} \otimes \ket{x}$ denote any state in the support of $\id - \Pi$ \emph{and} with records of size at most~$\abs{\set{j : x_j \neq \varnothing}} \leq t$. Notice that the latter condition is satisfied by the state~$(\id - \Pi) \psr{t}$ according to \Cref{Fact:rec}. Hence, it suffices to show that~$\norm{\Pi \rec \ket{\psi}} \leq \sqrt{\frac{10t}{m}} \norm{\ket{\psi}}$.

  Let $\ket{\psi} = \ps{\mathrm{id}} + \ps{\varnothing} + \sum_{y \in \Sigma} \ps{y}$ be the  decomposition of $\ket{\psi}$ as defined in the proof of \Cref{Prop:search}. We claim that $\norm{\Pi \rec \ps{\mathrm{id}}} = 0$, $\norm{\Pi \rec \ps{\varnothing}} \leq \sqrt{\frac{t}{m}} \norm{\ps{\varnothing}}$ and $\norm{\Pi \rec \ps{y}} \leq \frac{3t}{m} \norm{\ps{y}}$. The first statement is immediate, since once again $\rec \ps{\mathrm{id}} = \ps{\mathrm{id}}$. Let us detail the second statement (the last one is similar). By \Cref{Thm:recording}, the state evolves into
    \begin{align*}
      \Pi \rec \ps{\varnothing}
        & = \Pi \sum_{i,b,x : x_i = \varnothing, b \neq 0} \alpha_{i,b,x} \ket{i,b} \otimes \ket{\dots x_{i-1}}\pt[\Big]{\frac{1}{\sqrt{m}} \sum_{x_i' \in \Sigma} \omega^{bx_i'}\ket{x_i'}}\ket{x_{i+1} \dots} \\
        & = \sum_{i,b,x : x_i = \varnothing, b \neq 0} \alpha_{i,b,x} \ket{i,b} \otimes \ket{\dots x_{i-1}}\pt[\Big]{\frac{1}{\sqrt{m}} \sum_{\substack{x_i' \in \Sigma, \\ \exists j \neq i, x_j = x_i'}} \omega^{bx_i'}\ket{x_i'}}\ket{x_{i+1} \dots}.
     \end{align*}
   The norm is at most $\norm{\Pi \rec \ps{\varnothing}}^2 = \sum_{i,b,x : x_i = \varnothing, b \neq 0} \abs{\alpha_{i,b,x}}^2 \frac{\abs{\set{x_i' \in \Sigma: \exists j, x_j = x_i'}}}{m} \leq \frac{t}{m} \norm{\ps{\varnothing}}^2$ since the records with non-zero amplitudes $\alpha_{i,b,x} \neq 0$ are of size at most~$t$ by assumption. Finally, using the triangle and Cauchy--Schwarz inequalities, we have that $\norm{\Pi \rec \ket{\psi}} \leq \sqrt{\frac{10t}{m}} \norm{\ket{\psi}}$.
\end{proof}

We leave it to the reader to prove that any algorithm solving \colf\ using~$T$ queries must satisfy the final condition $\psucc^{\unif} = \pt{\Delta_T + \bo{\sqrt{T/m}}}^2 = \bo{T^3/m}$ (our definition of the computational model should be slightly adapted to allow for the output of two indices). This entails that the average-case quantum complexity must be at least~$Q_{\unif}(\colf) = \om{m^{1/3}}$, which is optimal since it matches the complexity of the existing quantum algorithms~\cite{BHT98c,Amb07j}.

%% file: adversary.tex
The \emph{adversary method} is arguably the most popular and versatile technique for proving quantum query lower bounds. It comes in many flavors, the simplest of which is the hybrid method presented in \Cref{Sec:hybrid}. The most evolved versions have virtually no limits, as they can always provide the optimal complexity (the catch being the difficulty in applying such methods to concrete problems). In this section, we will present the modern formulation of the adversary method, based on the spectral properties of an adversially chosen matrix with real-weight entries~\cite{HLS07c}.


\subsection{Technique}

Our presentation of the adversary method extends the approaches introduced in the hybrid and recording methods. Unlike in the previous section, we revert to the model with a Boolean input alphabet $x \in \rn^n$, Boolean decision problems $\rf$, and a joint binary oracle $\ora \pt*{\ket{i,b} \otimes \ket{x}} = \ket{i,b \oplus x_i} \otimes \ket{x}$.

First, the adversary method introduces the possibility of assigning weights to the pairs of inputs considered in the hybrid method, as follows.

\fakeparagraph{Weighted Gram matrix.}
Recall that the inner product $\ip{\psi_x^T}{\psi_y^T}$ between two final states with~$f(x) \neq f(y)$ relates to the probability with which the algorithm can be correct (final condition in \Cref{Thm:hybrid}). The adversary method exploits the entire Gram matrix $G_t = \pt{\ip{\psi_x^t}{\psi_y^t}}_{x,y \in \rn^n}$ and assigns weights~$\Gamma_{x,y}$ to its entries (a somewhat unintuitive aspect of the method is that it allows for negative weights as well). The inner products~$\ip{\psi_x^t}{\psi_y^t}$ evolve as queries are made to indices~$i$ with~$x_i \neq y_i$ (progress evolution in \Cref{Thm:hybrid}). The adversary method quantifies this evolution using the ``punctured'' weights $(\Gamma_i)_{x,y}$ obtained by zeroing out all entries of~$\Gamma$ with~$x_i = y_i$. The constraints on the resulting matrices are summarized in the following definition.

\begin{definition}[Adversary matrix]
  \label{Def:adversary}
  Let $\rf$ be a Boolean function over~$n$ variables. We say that $\Gamma \in \R^{2^n \times 2^n}$ is an \emph{adversary matrix} for~$f$ if it satisfies the two conditions,
  \begin{itemize}
    \item (Symmetric) $\Gamma_{x,y} = \Gamma_{y,x}$ for all $x,y \in \rn^n$,
    \item ($f$-sparse) $\Gamma_{x,y} = 0$ when $f(x) = f(y)$.
  \end{itemize}
  For all $i \in \ind$, we define the \emph{$i$-th punctured adversary matrix} $\Gamma_i \in \R^{2^n \times 2^n}$ of $\Gamma$ as the matrix satisfying also the condition,
  \begin{itemize}
    \item (Punctured) $(\Gamma_i)_{x,y} = 0$ when $x_i = y_i$, and $(\Gamma_i)_{x,y} = \Gamma_{x,y}$ otherwise.
  \end{itemize}
\end{definition}

While the adversary matrix provides a way to emphasize pairs of inputs that are difficult to distinguish, another source of hardness can be introduced by placing a distribution on the input, as was done in the recording method. This is achieved below by extending the purification technique introduced in the last section.

\fakeparagraph{Generalized input purification.}
We described in \Cref{Def:joint} how the state of an algorithm operating under an input distribution~$\dis$ can be represented as the bipartite state $\pst{t} = \sum_{x \in \rn^n} \sqrt{\dis(x)} \psx{t}\otimes\ket{x}$. The amplitude~$\sqrt{\dis(x)}$ can, in fact, be replaced with any complex number~$a_x$ satisfying $\abs{a_x}^2 = \dis(x)$, since the reduced density matrix \emph{of the algorithm} remains equal to $\rho^t = \sum_{x \in \Sigma^n} \dis(x) \proj{\psi_x^t}$. However, this modification can affect the reduced density matrix \emph{of the input register} $\varrho^t = \sum_{x,y} a_x a_y^* \ip{\psi_x^t}{\psi_y^t} \op{x}{y}$. In the adversary method, the choice of these amplitudes determines the initial value~$\Delta_0$ of the progress measure, which we aim to minimize in order to allow for a longer progress evolution. The progress is going to be quantified by the expression
  \[- \Tr(\Gamma\varrho^t) = - \abs{\braket{\psi^{t}}{(\id \otimes \Gamma)}{\psi^{t}}} = - \abs[\Big]{\sum_{x,y} \Gamma_{x,y} a_x a_y^* \ip{\psi_x^t}{\psi_y^t}},\]
which is minimized at $t = 0$ when $(a_x)_x$ is chosen as the principal eigenvector of the adversary matrix~$\Gamma$. With this choice, and by applying an offset of $\norm{\Gamma}$ to initialize the progress at~$0$, the adversary method is formulated as follows.

\begin{theorem}[Adversary method]
  \label{Thm:adversary}
  Consider a quantum algorithm that accesses a random input~$x \in \rn^n$ drawn from a distribution~$\dis$. Let $a \in \C^{2^n}$ be a unit vector such that $\dis(x) = \abs{a_x}^2$ for all~$x$. Define the joint state of the algorithm and input after~$t$ queries as,
    \[\pst{t} = \sum_{x \in \rn^n} a_x \psx{t}\otimes\ket{x}.\]
  Given a function~$\rf$ and a non-zero adversary matrix~$\Gamma \in \R^{2^n \times 2^n}$ for~$f$ with principal eigenvector~$a$, define the following progress measure,
    \[\Delta_t = \norm{\Gamma} - \abs{\braket{\psi^{t}}{(\id \otimes \Gamma)}{\psi^{t}}}\]
  Then the progress obeys the following inequalities,
  \begin{itemize}
    \item (Initial condition) $\Delta_0 = 0$,
    \item (Progress evolution) $\Delta_{t+1} \leq \Delta_t + 2 \max_{i \in \ind} \norm{\Gamma_i}$.
  \end{itemize}
    Furthermore, the average success probability $\psucc^{\dis}$ of the algorithm in computing~$f$ after $T$ queries is at most,
  \begin{itemize}
    \item (Final condition) $\psucc^{\dis} \leq \frac{1}{2} + \sqrt{\Delta_T / (2\norm{\Gamma})}$.
  \end{itemize}
  Consequently, the average-case quantum query complexity of $f$ under the distribution~$\dis$ is at least $Q_{\dis}(f) \geq \frac{\norm{\Gamma}}{36 \max_{i \in \ind} \norm{\Gamma_i}}$.
\end{theorem}

\begin{proof}[Proof of the initial condition]
  The initial state is $\pst{0} = \ket{0,0} \otimes \ket{a}$ where $\ket{a} = \sum_{x \in \rn^n} a_x \ket{x}$ is a principal unit eigenvector of~$\Gamma$. Hence, $\Delta_0 = \norm{\Gamma} - \abs{\braket{a}{\Gamma}{a}} = 0$.
\end{proof}

\begin{proof}[Proof of the progress evolution]
  The progress increase is bounded by $\Delta_{t+1} - \Delta_t \leq \abs{\braket{\psi^{t+1}}{(\id \otimes \Gamma)}{\psi^{t+1}} - \braket{\psi^{t}}{(\id \otimes \Gamma)}{\psi^{t}}}$. By definition, $\pst{t+1} = (U_{t+1} \otimes \id)\ora \pst{t}$ where $\ora \pt*{\ket{i,b} \otimes \ket{x}} = \pt*{\ora_x \ket{i,b}} \otimes \ket{x} = \ket{i,b \oplus x_i}\otimes\ket{x}$ is the joint binary oracle. Since $(U_{t+1} \otimes \id)$ and $(\id \otimes \Gamma)$ commute, we obtain that
    \[\Delta_{t+1} - \Delta_t \leq \abs{\braket{\psi^{t}}{\ora (\id \otimes \Gamma)\ora}{\psi^{t}} - \braket{\psi^{t}}{(\id \otimes \Gamma)}{\psi^{t}}}.\]
  The operator $\ora (\id \otimes \Gamma)\ora - \id \otimes \Gamma$ can be expressed in the standard basis as $\sum_{i,b,x,y} \pt{\op{i,b\oplus x_i}{i,b\oplus y_i} - \proj{i,b}} \otimes \Gamma_{x,y} \op{x}{y}$. Its action on the value register is represented by the two-by-two matrix $\sum_{b \in \rn} \op{b\oplus x_i}{b\oplus y_i} - \proj{b}$, which equals $0$ when $x_i = y_i$, and $X - \id$ when $x_i \neq y_i$ (where $X$ is the Pauli-$X$ matrix). The condition on $(x_i,y_i)$ can be absorbed into the matrix $\Gamma_i$, since $(\Gamma_i)_{xy} = 0$ when $x_i = y_i$ by definition (punctured property). Hence, we can continue the above inequality as follows,
    \begin{align*}
      \Delta_{t+1} - \Delta_t
        & \leq \abs[\Big]{\sum_{i \in \ind} \braket{\psi^{t}}{\pt*{\proj{i} \otimes (X - \id) \otimes \Gamma_i}}{\psi^{t}} }\\
        & \leq \sum_{i \in \ind} \norm{\id \otimes (X - \id) \otimes \Gamma_i} \cdot \norm{(\proj{i} \otimes \id) \pst{t}}^2  \\ \intertext{\raggedleft (Cauchy-Schwarz inequality)}
        & \leq \max_{i \in \ind} \norm{\id \otimes (X - \id) \otimes \Gamma_i} \cdot \sum_{i \in \ind} \norm{(\proj{i} \otimes \id) \pst{t}}^2 \\
        &  = 2 \max_{i \in \ind} \norm{\Gamma_i}. \qedhere
    \end{align*}
\end{proof}

\begin{proof}[Proof of the final condition]
  The average success probability of the algorithm is defined as $\psucc^{\dis} = \norm{\prsucc \pst{T}}^2$, where $\prsucc$ is the projector onto the states $\ket{i,b} \otimes \ket{x}$ with $b = f(x)$. The progress after~$T$ queries satisfies the equality:
  \switchAMS{
    \begin{align*}
      & \norm{\Gamma} - \Delta_T \\
        & \quad = \abs{\braket{\psi^{T}}{(\id \otimes \Gamma)}{\psi^{T}}} \\
        & \quad = \bigl|\braket{\psi^T}{\prsucc(\Gamma \otimes \id)\prsucc}{\psi^T} + \braket{\psi^T}{(\id-\prsucc)(\Gamma \otimes \id)(\id-\prsucc)}{\psi^T} \\
        & \quad \quad + \braket{\psi^T}{\prsucc(\Gamma \otimes \id)(\id - \prsucc)}{\psi^T} + \braket{\psi^T}{(\id-\prsucc)(\Gamma \otimes \id)\prsucc}{\psi^T}\bigr| \\
        & \quad = \abs{\braket{\psi^T}{\prsucc(\Gamma \otimes \id)(\id - \prsucc)}{\psi^T} + \braket{\psi^T}{(\id-\prsucc)(\Gamma \otimes \id)\prsucc}{\psi^T}} 
    \end{align*}
  }{
    \begin{align*}
      \norm{\Gamma} - \Delta_T
        & = \abs{\braket{\psi^{T}}{(\id \otimes \Gamma)}{\psi^{T}}} \\
        & = \bigl|\braket{\psi^T}{\prsucc(\Gamma \otimes \id)\prsucc}{\psi^T} + \braket{\psi^T}{(\id-\prsucc)(\Gamma \otimes \id)(\id-\prsucc)}{\psi^T} \\
        & \quad + \braket{\psi^T}{\prsucc(\Gamma \otimes \id)(\id - \prsucc)}{\psi^T} + \braket{\psi^T}{(\id-\prsucc)(\Gamma \otimes \id)\prsucc}{\psi^T}\bigr| \\
        & = \abs{\braket{\psi^T}{\prsucc(\Gamma \otimes \id)(\id - \prsucc)}{\psi^T} + \braket{\psi^T}{(\id-\prsucc)(\Gamma \otimes \id)\prsucc}{\psi^T}}
    \end{align*}
  }
  where the last equality uses that $\Gamma_{xy} = 0$ when $f(x) = f(y)$ ($f$-sparse property). By the Cauchy-Schwarz inequality, we obtain
  \begin{align*}
    \norm{\Gamma} - \Delta_T
      & \leq 2 \norm{\Gamma \otimes \id} \cdot \norm{\prsucc\ket{\psi^T}} \cdot \norm{(\id - \prsucc)\ket{\psi^T}} \\
      & = 2 \norm{\Gamma} \cdot \sqrt{\psucc^{\dis} (1 - \psucc^{\dis})} \\
      & \leq 2 \norm{\Gamma} \cdot (1/4 + \psucc^{\dis} (1 - \psucc^{\dis})) = 2 \norm{\Gamma} \cdot \pt*{1/2-(\psucc^{\dis} - 1/2)^2}.
  \end{align*}
  Reordering the terms, we obtain that $\psucc^{\dis} \leq 1/2 + \sqrt{\Delta_T/(2\norm{\Gamma})}$.
  Using the bound on the progress evolution, the success probability after $T$ queries must be at most $\psucc^{\dis} \leq 1/2 + \sqrt{T \cdot \max_{i \in \ind} \norm{\Gamma_i}/\norm{\Gamma}}$, hence the algorithm requires at least $T \geq \frac{\norm{\Gamma}}{36 \max_{i \in \ind} \norm{\Gamma_i}}$ queries to succeed with probability~$\psucc^{\dis} \geq 2/3$.
\end{proof}

The adversary method is often stated under the worst-case output condition, where the algorithm must be correct on \emph{any} input with probability at least~$2/3$ (as was defined in \Cref{Sec:intro} and \Cref{Def:qc}). In this case, the lower bound depends on the \emph{adversary value}~$\adv(f)$, which is obtained by maximizing the above argument over the entire set of adversary matrices (or, equivalently, over all input distributions).

\begin{definition}[Adversary value]
  \label{Def:adversaryValue}
  The \emph{adversary value} of a function~$\rf$ is the non-negative real number defined as,
    \[\adv(f) = \max_{\Gamma} \frac{\norm{\Gamma}}{\max_{i \in \ind} \norm{\Gamma_i}}\]
  where the maximum is taken over all non-zero adversary matrices $\Gamma \in \R^{2^n \times 2^n}$ for~$f$.
\end{definition}

\begin{corollary}
  \label{Cor:adversary}
  The quantum query complexity of any function $\rf$ is at least~$Q(f) \geq \adv(f)/36$.
\end{corollary}

Unlike other combinatorial measures of complexity -- such as the block sensitivity -- the adversary value is always equal to the optimal query complexity, up to a constant factor. The proof of the upper bound $Q(f) = \bo{\adv(f)}$ will be the focus of \Cref{Sec:dual}. The crucial insight is that~$\adv(f)$ can be expressed as the optimum of a semidefinite program, whose dual can be turned into a quantum algorithm.

Another interesting aspect of the adversary method is that it naturally scales with function composition, in the sense that applying it to a composed function $f \bullet g$ only requires understanding the method separately for $f$ and $g$.

\begin{proposition}[Function composition \cite{HLS07c,Rei09c}]
  \label{Prop:composition}
  Let $\rf$ and $g : \rn^m \ra \rn$ be two Boolean functions. Define their composition $f \bullet g : \rn^{nm} \ra \rn$ as $f \bullet g(x) = f(g(x_{1},\dots,x_{m}),\dots,g(x_{(n-1)m+1},\dots,x_{nm}))$. Then, $\adv(f \bullet g) = \adv(f)\adv(g)$.
\end{proposition}

The reader interested in the proof of this result may attempt to demonstrate the case $\adv(\orf \bullet g) \geq \sqrt{n} \cdot \adv(g)$ by combining an optimal adversary matrix for~$g$ with the adversary matrix for the \orf\ function described in the next section.


\subsection{Applications}

One of the most challenging aspects of the adversary method is identifying which adversary matrices maximize the ratio $\norm{\Gamma}/(\max_i \norm{\Gamma_i})$. In this section, we develop some applications where~$\Gamma$ belongs to~$\rn^{2^n \times 2^n}$ (sometimes referred to as the ``basic adversary method''). This restriction on the entries of~$\Gamma$ can severely limit the adversary method in general, but it makes it somewhat more intuitive. The merits of introducing more general classes of matrices were identified later (see~\cite{SS06ja} for instance), with the most general case -- allowing negative entries -- discovered by H{\o}yer, Lee and {\v{S}}palek~\cite{HLS07c}.


\fakeparagraph{Application 1: The \orf\ function.}
We present a simple adversary matrix for the \orf\ function, demonstrating that its query complexity is at least~$\om{\sqrt{n}}$. Note that the constant factor in this lower bound is slightly smaller than those established in \Cref{Prop:hybridOR,Prop:ordeg}.

\begin{proposition}[Adversary method applied to \orf]
  \label{Prop:advOR}
  The quantum query complexity of the \orf\ function is at least~$Q(\orf) \geq \sqrt{n}/36$.
\end{proposition}

\begin{proof}
  We consider the $(0,1)$-adversary matrix~$\Gamma$ with nonzero weights placed on the same hard-to-distinguish pairs of inputs $(x\super{0},y\super{i})$ as defined in the application of the hybrid method (\Cref{Prop:hybridOR}). The non-zero parts of the adversary and punctured adversary matrices are represented in \Cref{Fig:advMat}.

  \begin{figure}
    \[
      \Gamma =
      \begin{array}{c c}
      \begin{blockarray}{c@{\hspace{2mm}}c@{\hspace{2mm}}c@{\hspace{2mm}}cc}
        x\super{0} & y\super{1} & \dots & y\super{n} \\
        \begin{block}{(c@{\hspace{2mm}}c@{\hspace{2mm}}c@{\hspace{2mm}}c)c}
          0 & 1 & \dots & 1 & x\super{0} \\
          1 & 0 & \dots & 0 & y\super{1} \\
          \vdots & \vdots & \ddots & \vdots & \vdots \\
          1 & 0 & \dots & 0 & y\super{n} \\
        \end{block}
      \end{blockarray}
      &
      \Gamma_i =
      \begin{blockarray}{c@{\hspace{2mm}}c@{\hspace{2mm}}c@{\hspace{2mm}}c@{\hspace{2mm}}c@{\hspace{2mm}}c@{\hspace{2mm}}cc}
        x\super{0} & y\super{1} & \dots & y\super{i} & \dots & \dots & y\super{n} \\
        \begin{block}{(c@{\hspace{2mm}}c@{\hspace{2mm}}c@{\hspace{2mm}}c@{\hspace{2mm}}c@{\hspace{2mm}}c@{\hspace{2mm}}c)c}
          0 & \dots & 0 & 1 & 0 & \dots & 0 & x\super{0} \\
          \vdots & 0 & \dots & \dots & \dots & \dots & 0 & y\super{1} \\
          0 & \vdots & \ddots & & & & \vdots & \vdots \\
          1 & \vdots & & \ddots & & & \vdots & y\super{i} \\
          0 & \vdots & & & \ddots & & \vdots & \vdots \\
          \vdots & \vdots &  &  &  & \ddots & \vdots & \vdots \\
          0 & 0 & \dots & \dots & \dots & \dots & 0 & y\super{n} \\
        \end{block}
      \end{blockarray}
      \end{array}
      \]
    \caption{Adversary matrix for the \orf\ function.}
    \label{Fig:advMat}
  \end{figure}

  It is a simple calculation to check that $\norm{\Gamma} = \sqrt{n}$ and $\norm{\Gamma_i} = 1$, hence $Q(\orf) \geq \norm{\Gamma}/(36 \norm{\Gamma_i}) = \sqrt{n}/36$. The fact that $\pt[\big]{\frac{1}{\sqrt{2}},\frac{1}{\sqrt{2n}},\dots,\frac{1}{\sqrt{2n}}}$ is a principal eigenvector of~$\Gamma$ informs us that the lower bound holds also for the average-case complexity $Q_{\dis}(\orf) \geq \sqrt{n}/36$ under the input distribution $\dis(x\super{0}) = 1/2$, $\dis(y\super{1}) = \dots = \dis(y\super{n}) = 1/(2n)$.
\end{proof}


\fakeparagraph{Application 2: Combinatorial formulation of the basic adversary method.}
The adversary method was originally introduced by Ambainis~\cite{Amb02j} for adversary matrices with~$(0,1)$ entries. Here, we reproduce the combinatorial formulation of this approach and develop its application to the \conf\ problem in the next section. The basic adversary method represents hard-to-distinguish pairs of inputs as a bipartite graph and relates the adversary value (and, therefore, the query complexity) to the degree expansion properties of that graph.

\begin{proposition}[Basic adversary method \cite{Amb02j}]
  \label{Prop:basicAdv}
  Let $\rf$ be a Boolean function. Choose two sets of inputs $V_0 \subseteq \set{x : f(x) = 0}$, $V_1 \subseteq \set{x : f(x) = 1}$ and a relation $E \subseteq V_0 \times V_1$. Let~$G$ be the bipartite graph with vertex set $V = V_0 \cup V_1$ and edge set $E$. For each index $i \in \ind$, define the subgraph $G_i$ of~$G$ obtained by removing all edges $(x,y)$ for which~$x_i = y_i$.

  For each $x \in \rn^n$, let $d(x)$ be the degree of the vertex $x$ in the graph $G$ and $d(x,i)$ be its degree in the graph $G_i$. Then, the adversary value of $f$ is at least
    \[\adv(f) \geq \sqrt{\frac{\min_{(x,y) \in V_0 \times V_1} d(x) d(y)}{\max_{(x,y) \in E, i \in \ind} d(x,i) d(y,i)}}.\]
\end{proposition}

\begin{proof}
  We define~$\Gamma$ as the adjacency matrix of the graph~$G$ (i.e., $\Gamma_{x,y} = 1$ if and only if $(x,y) \in E$). One can check that $\Gamma$ is indeed an adversary matrix. Let $m_0 = \min_{x \in V_0} d(x)$ and $m_1 = \min_{y \in V_1} d(y)$ be the minimal left and right degrees in the bipartite graph $G$. We can assume without loss of generality that $m_0,m_1 \geq 1$ (otherwise the lower bound is vacuous).

  We show first that $\norm{\Gamma} \geq \sqrt{m_0 m_1}$. Let $a \in \R^n$ be the vector defined as $a_x = \sqrt{m_0/(2m_1\abs{E})}$ when $x \in V_0$ and $a_y = \sqrt{m_1/(2m_0\abs{E})}$ when $y \in V_1$. Then,
    $\norm{a}^2 = m_0\abs{V_0}/(2m_1\abs{E}) + m_1\abs{V_1}/(2m_0\abs{E}) \leq 1$
  and
    $\norm{\Gamma a}^2
    = \sum_{x \in V_0} d(x)^2 m_1/(2m_0\abs{E}) + \sum_{y \in V_1} d(y)^2 m_0/(2m_1\abs{E})
    \geq m_0 m_1 \abs{V_0}/(2\abs{E}) + m_0 m_1 \abs{V_1}/(2\abs{E})
    = m_0m_1$.
  Hence, $\norm{\Gamma} \geq \norm{\Gamma a}/\norm{a} \geq \sqrt{m_0 m_1}$.

  We claim next that $\norm{\Gamma_i} \leq \max_{(x,y) \in E} \sqrt{d(x,i) d(y,i)}$ for all $i$. This follows immediately from the next inequality on the spectral norm of symmetric $(0,1)$-matrices.
    \begin{lemma}[Appendix A in \cite{SS06ja}]
      The spectral norm of a symmetric $(0,1)$-matrix~$A$ is at most $\norm{A} \leq \max_{x,y : A_{x,y} = 1} \sqrt{r_x(A)c_y(A)}$, where $r_x(A)$ (resp. $c_y(A)$) is the sum of the elements in the $x$-th row (resp. $y$-th column) of $A$.
    \end{lemma}

  The proposition follows by definition of $\adv(f) \geq \sqrt{\norm{\Gamma}/\max_{i \in \ind} \norm{\Gamma_i}}$.
\end{proof}


\fakeparagraph{Application 3: The \conf\ function.}
We study the complexity of determining whether an undirected $n$-vertex graph is connected or not. The input is encoded over $\binom{n}{2}$ bits~$x \in \rn^{\binom{n}{2}}$ representing the adjacency matrix of the graph (where $x_{\set{i,j}} = 1$ if and only if there is an edge between vertices~$i$ and~$j$). This input is \emph{not} meant to represent the graphs~$G$ introduced in the basic adversary method (\Cref{Prop:basicAdv}), which are distinct objects used to analyze the adversary value. 

\begin{definition}[\conf]
  The $\conf$ problem is to output~$1$ if the graph represented by the input~$x \in \rn^{\binom{n}{2}}$ is connected, and~$0$ otherwise.
\end{definition}

The query complexity of this problem was first established by D{\"{u}}rr, Heiligman, H{\o}yer and Mhalla~\cite{DHHM06j}. The classical query complexity is easily shown to be maximal through a reduction from the \orf\ problem.

\begin{proposition}
  The randomized query complexity of the \conf\ function is at least~$R(\conf) = \om{n^2}$.
\end{proposition}

\begin{proof}
  Assume that~$n$ is even. We describe a reduction from the \orf\ problem over $n^2/4$ bits to the \conf\ problem over $\binom{n}{2}$ bits. Fix $P_1$ and $P_2$ to be two path graphs of lengths $n/2$ over the vertices~$\set{1,\dots,n/2}$ and~$\set{n/2+1,\dots,n}$, respectively. Given an input $y \in \rn^{n/2 \times n/2}$ to the \orf\ problem (indexed as a square matrix for convenience), define the graph obtained by the union of $P_1$, $P_2$ and all the edges $\set{i,j+n/2}$ for which $y_{i,j} = 1$. Then, the resulting graph is connected if and only if $\orf(y) = 1$. Since computing~$\orf(y)$ requires~$\om{n^2}$ queries, the reduction implies that the same lower bound applies to the \conf\ problem as well.
\end{proof}

The same reduction yields an~$\om{\sqrt{n^2/4}} = \om{n}$ quantum lower bound for \conf. This is, however, not optimal, as shown in the next proposition based on the basic adversary method (we will describe a matching upper bound in \Cref{Prop:conUB}, using the dual to the adversary method).

\begin{proposition}[Adversary method applied to \conf]
  \label{Prop:conLB}
  The quantum query complexity of the \conf\ function is at least~$Q(\conf) = \om{n^{3/2}}$.
\end{proposition}

\begin{proof}
  We first construct the two sets of inputs~$V_0$ and $V_1$ needed to apply \Cref{Prop:basicAdv}. The set~$V_0$ consists of the graphs $x \in \rn^{\binom{n}{2}}$ made of two disjoint cycles, each of length at least~$n/3$. The set~$V_1$ consists of the graphs $x \in \rn^{\binom{n}{2}}$ made of a single cycle of length~$n$. We define the relation $E \subseteq V_0 \times V_1$ as all pairs of graphs $(x,y) \in V_0 \times V_1$ that are related by the following process: $y$ can be obtained by disconnecting one edge from each cycle in~$x$ and gluing the two resulting paths together into a cycle of length~$n$.

  Using the notations of \Cref{Prop:basicAdv}, each input $x \in V_0$ belongs to at least $d(x) \geq (n/3)^2$ pairs in $E$, since there are at least $n/3$ choices to disconnect each cycle in~$x$. Conversely, each input $y \in V_1$ belongs to at least $d(y) \geq n^2/6$ pairs since there are at least $n^2/6$ choices to disconnect two edges in $y$ at distance at least $n/3$ from each other. We leave it to the reader to verify that the number of pairs remaining, when an edge $\set{i,j}$ must be in one graph but not in the other, satisfies the relation $d(x,\set{i,j}) d(y,\set{i,j}) = \bo{n}$. By \Cref{Prop:basicAdv} and \Cref{Cor:adversary}, the quantum query complexity is at least $Q(\conf) = \om{\adv(\conf)} = \om{\sqrt{n^2 \cdot n^2 / n}} = \om{n^{3/2}}$.
\end{proof}

\conf\ is part of a larger family of graph problems -- the non-trivial monotone graph properties -- whose complexities attract a lot of attention. The interested reader can refer to the Aanderaa-Karp-Rosenberg conjectures (e.g., \cite[Section 5]{ABK21c}).

%% file: dual.tex
The study of lower-bound methods is often not entirely separate from that of upper bounds, i.e., algorithm design. This principle is particularly well illustrated by the adversary method. Indeed, any solution to its dual, defined through semidefinite optimization, can be adapted into a surprisingly efficient converse algorithm. This result, first established by Reichardt~\cite{Rei11c}, provides a tight characterization of quantum query complexity in terms of the adversary value, up to constant factors:~$Q(f) = \ta{\adv(f)}$. This stands in contrast to the dual of the polynomial method, as described in \Cref{Prop:dualpoly}, which does not have an algorithmic interpretation in general (although this can be addressed by a recent extension of the polynomial method~\cite{ABP19j}).

This section presents the dual of the adversary method and explains how its solutions can be converted into quantum algorithms. We follow the approach of~\cite{LMR11c,dWol19p}. In the application section, we derive optimal algorithms for the \orf, \aotf\ and \conf\ functions.


\subsection{Technique}

The adversary value $\adv(f)$ was defined in the previous section as the supremum of the ratio $\norm{\Gamma}/\max_{i \in \ind}\norm{\Gamma_i}$, evaluated over the set of adversary matrices~$\Gamma$ (\Cref{Def:adversaryValue}). It is not complicated to see that this optimization problem can be phrased as a semidefinite program~(SDP), i.e., a generalization of a linear program where, in addition to linear constraints, some variables~$X$ (structured as square matrices) must satisfy the positive semidefinite constraint (PSD)~$X \succeq 0$.

The dual of an SDP often provides valuable insights into its optimal solutions. Below, we state the primal-dual formulation of the adversary value. The dual program associates a PSD matrix~$V\super{i}$ with each query index $i$, and it minimizes the largest diagonal entry of~$\sum_i V\super{i}$ under the constraint that the off-diagonal terms of the partial sum $\sum_{i : x_i \neq y_i} V\super{i}_{x,y}$ are equal to $1$ whenever~$f(x) \neq f(y)$.

\begin{proposition}[Dual program for $\adv(f)$, proven in \cite{Rei09c} or Theorem 3.29 in \cite{Bel14p}]
  \label{Prop:dual}
  Let $\rf$ be a Boolean function. Define the two following semidefinite programs.
  \switchAMS{
  \begin{gather*}
    \textit{Primal semidefinite program}  \\
    \boxed{\begin{array}{ll}
      \max\limits_{\Gamma}  &  \frac{\norm{\Gamma}}{\max_{i \in \ind}\norm{\Gamma_i}}    \\[2mm]
      \mbox{s.t.} & \Gamma_{x,y} = \Gamma_{y,x} \hfill \forall x,y \in \rn^n \\[2mm]
      & \Gamma_{x,y} = 0 \quad \forall x,y: f(x) = f(y) \\[2mm]
      & \Gamma \in \R^{2^n \times 2^n}
    \end{array}} \\[5mm]
     \textit{Dual semidefinite program} \\
     \boxed{\begin{array}{ll}
      \min\limits_{V\super{1},\dots,V\super{n}} &  \max\limits_{x \in \rn^n}  \sum\limits_{i \in \ind} V\super{i}_{x,x} \\[2mm]
      \mbox{s.t.} & V\super{i} \succeq 0 \hfill \forall i \in \ind  \\[2mm]
      & \sum_{i : x_i \neq y_i} V\super{i}_{x,y} = 1 \quad \forall x,y: f(x) \neq f(y)\\[2mm]
      & V\super{1},\dots,V\super{n} \in \C^{2^n \times 2^n}
    \end{array}}
  \end{gather*}
  }{
    \[\begin{array}{@{}cc@{}}
          \textit{Primal semidefinite program} & \textit{Dual semidefinite program} \\
          \boxed{\begin{array}{ll}
            \max\limits_{\Gamma}  &  \frac{\norm{\Gamma}}{\max_{i \in \ind}\norm{\Gamma_i}}    \\[2mm]
            \mbox{s.t.} & \Gamma_{x,y} = \Gamma_{y,x} \hfill \forall x,y \in \rn^n \\[2mm]
            & \Gamma_{x,y} = 0 \quad \forall x,y: f(x) = f(y) \\[2mm]
            & \Gamma \in \R^{2^n \times 2^n}
          \end{array}}
          &
          \boxed{\begin{array}{ll}
            \min\limits_{V\super{1},\dots,V\super{n}} &  \max\limits_{x \in \rn^n}  \sum\limits_{i \in \ind} V\super{i}_{x,x} \\[2mm]
            \mbox{s.t.} & V\super{i} \succeq 0 \hfill \forall i \in \ind  \\[2mm]
            & \sum_{i : x_i \neq y_i} V\super{i}_{x,y} = 1 \quad \forall x,y: f(x) \neq f(y)\\[2mm]
            & V\super{1},\dots,V\super{n} \in \C^{2^n \times 2^n}
          \end{array}}
      \end{array}\]
   }

  Then, the two programs are dual to each other, and they satisfy the strong duality property. In particular, their optimum are both equal to $\adv(f)$.
\end{proposition}

By (weak) duality, one can certify that the adversary value is at most $\adv(f) \leq T$ by exhibiting a feasible solution to the dual program with value $T$. The striking property, which we establish next, is the ability to also derive a quantum algorithm for computing~$f$ with complexity~$\bo{T}$. The rest of this section is dedicated to the description and analysis of this algorithm. We start by rewriting the solutions of the dual program in terms of the vector realizations of the PSD matrices.

\begin{lemma}[Vector realization of a dual solution]
  \label{Lem:rep}
  Let $V\super{1},\dots,V\super{n} \in \C^{2^n \times 2^n}$ be a feasible solution to the dual program from \Cref{Prop:dual} with value $T = \max_{x \in \rn^n} \sum_{i \in \ind} V\super{i}_{x,x}$. Then, there exist an integer~$d$ and a set of complex vectors $\set{\ket{w\super{x,i}}}_{x \in \rn^n, i \in \ind} \subset \C^d$ such that
  $\sum_{i : x_i \neq y_i} \ip{w\super{x,i}}{w\super{y,i}} = 1$ when $f(x) \neq f(y)$, and $T = \max_{x \in \rn^n} \sum_{i \in \ind} \norm{w\super{x,i}}^2$.
\end{lemma}

\begin{proof}
  A necessary (and sufficient) condition for a matrix $V \in \C^{2^n \times 2^n}$ to be PSD is to be equal to a Gram matrix $V_{x,y} = \ip{w\super{x}}{w\super{y}}$ for some vectors $\ket{w\super{0}},\dots,\ket{w\super{2^n-1}} \in \C^d$ and integer~$d$. It is immediate to verify that using such a vector realization for the matrices $V\super{1},\dots,V\super{n}$ leads to the conclusion of the lemma.
\end{proof}

The dimension~$d$ of a vector solution can always be chosen as $2^n$ (for instance, by taking the column vectors of the Cholesky decompositions of the matrices $V\super{1},\dots,V\super{n}$). However, the smaller~$d$ is, the less memory the quantum algorithm will require.

The algorithm is going to determine the value of~$f(x)$ based on the outcome of a quantum phase estimation procedure. The estimated eigenphase will be close to~$0$ when~$f(x) = 1$ and at least~$1/(2T)$ when~$f(x) = 0$. The eigenphase gap between the two cases ensures that running phase estimation with a precision of~$\bo{1/T}$ is sufficient.

We now define the input unitary $R_x$ used in the phase estimation procedure. It is given by the product of two reflection operators acting on three registers, over a Hilbert space~$\Hil$ of dimension~$2nd+1$ (in particular, the algorithm is not memoryless). The first two registers are the index and value registers. The third register has~$d$ dimensions and is used to encode the given vector realization. The extra dimension allows us to define a special state~$\pss$, which will serve as the guiding state for phase estimation.

\begin{definition}[Unitaries $R_x$ associated with a vector realization]
  \label{Def:refl}
  Fix two integers $n,d$. For each $x \in \rn^n$, let~$\Hil_x$ be the Hilbert space of dimension $2nd+1$ defined as,
    \[\Hil_x = \spn\set[\big]{\ket{i,x_i} \otimes \ket{w} : i \in \ind, w \in \set{1,\dots,d}} \oplus \spn\set*{\ket{1,0} \otimes \ket{d+1}}.\]
  Let $\Hil = \sum_x \Hil_x$ be the Hilbert space spanned by the union of the sets. Define $\pss = \ket{1,0} \otimes \ket{d+1}$.

  Fix a vector realization $\ket{w\super{x,i}} \in \spn\set{\ket{1},\dots,\ket{d}}$ as described in \Cref{Lem:rep}. For each $x \in \rn^n$, let $\ket{t_x^+} \in \Hil_x$ and $\ket{t_x^-} \in \Hil_{\bar{x}}$ be the two (unnormalized) states defined as,
  \switchAMS{
    \begin{gather*}
      \ket{t_x^+} = \pss + \frac{1}{\sqrt{4T}} \sum_{i \in \ind} \ket{i,x_i} \otimes \ket{w\super{x,i}}, \\
      \ket{t_x^-} = \pss - \sqrt{4T} \sum_{i \in \ind} \ket{i,\br{x_i}} \otimes \ket{w\super{x,i}},
    \end{gather*}
  }{
    \[\ket{t_x^+} = \pss + \frac{1}{\sqrt{4T}} \sum_{i \in \ind} \ket{i,x_i} \otimes \ket{w\super{x,i}}
    \quad \text{and}\quad
    \ket{t_x^-} = \pss - \sqrt{4T} \sum_{i \in \ind} \ket{i,\br{x_i}} \otimes \ket{w\super{x,i}},\]
  }
  where $\br{x_i} = 1 - x_i$. 
  Finally, define the unitary~$R_x$ acting on~$\Hil$ as the product,
    \[R_x = \pt{2 \Pi_x - \id}\pt{2\Delta - \id}.\]
  where~$\Pi_x$ and $\Delta$ are the orthogonal projectors onto $\Hil_x$ and $\spn\set*{\ket{t_y^-} : f(y) = 1}$ respectively.
\end{definition}

Before analyzing the spectral properties of~$R_x$, note that it can be implemented using only two queries to a quantum oracle for~$x$.

\begin{lemma}[Query implementation of~$R_x$]
  \label{Lem:qimplem}
  There exist two unitary operators $U_0,U_1$ such that, for any~$x \in \rn^n$, the operation $R_x$ given in \Cref{Def:refl} can be decomposed as the product $R_x = \pt{\ora_x \otimes \id} U_1 \pt{\ora_x \otimes \id} U_0$, where $\ora_x$ is the quantum binary oracle to~$x$.
\end{lemma}

\begin{proof}
  The first unitary can be chosen as the reflection~$U_0 = 2\Delta - \id$, since the projector~$\Delta$ does not depend on~$x$. By definition, the second reflection~$2 \Pi_x - \id$ must flip the sign of the basis states~$\ket{i,\br{x_i}} \otimes \ket{w}$ when $w \neq d+1$, and keep the other states the same. Omitting the last register, it suffices to use the Pauli-Z gate to obtain $\ora_x (\id \otimes Z) \ora_x \ket{i,\br{x_i}} = - \ket{i,\br{x_i}}$ and $\ora_x (\id \otimes Z) \ora_x \ket{i,x_i} = \ket{i,x_i}$. The unitary~$U_1$ applies the same Pauli-Z transformation, but conditioned on the last register being not in the state $\ket{d+1}$.
\end{proof}

We now state the central result regarding the spectral properties of~$R_x$. If $f(x) = 1$ then it is shown that the state~$\pss$ lies close the eigenspaces of~$R_x$ with eigenphase~$0$. On the other hand, if~$f(x) = 0$ then the state~$\pss$ principally belongs to the eigenspaces with larger eigenphases~$\om{1/T}$.

\begin{proposition}[Phase gap]
  \label{Lem:epGap}
  Let $\Lambda_{x,\theta}$ be the projector onto the eigenspaces of $R_x$ with eigenvalues in the set $\set{e^{i\varphi} : \abs{\varphi} \leq \theta}$. Then, $\norm{\Lambda_{x,0} \pss}^2 \geq 3/4$ if $f(x) = 1$ and $\norm{\Lambda_{x,\frac{1}{3T}} \pss}^2 \leq 2/9$ if~$f(x) = 0$.
\end{proposition}

\begin{proof}
  First, we consider the case~$f(x) = 1$. The state~$\ket{t_x^+}$ is both in the support of~$\Delta$ and~$\Pi_x$. Hence, $R_x \ket{t_x^+} = \ket{t_x^+}$, meaning that it is in the support of~$\Lambda_{x,0}$ as well. The lemma follows by observing that the distance with the state~$\pss$ is at most $\norm{\pss - \ket{t_x^+}}^2 = \frac{1}{4T} \norm{\sum_{i \in \ind} \ket{i,x_i} \otimes \ket{w\super{x,i}}}^2 = \frac{1}{4T} \sum_{i \in \ind} \norm{\ket{w\super{x,i}}}^2 \leq \frac{1}{4T} \cdot T \leq \frac{1}{4}$.

  Next, we consider the case~$f(x) = 0$. We let the reader verifies that the state~$\ket{t_x^-}$ is orthogonal to the support of~$\Delta$, i.e., $\Delta \ket{t_x^-} = 0$, and its projection onto the support of~$\Pi_x$ is $\pss = \Pi_x \ket{t_x^-}$ (the former equality makes use of the condition $\sum_{i : x_i \neq y_i} \ip{w\super{x,i}}{w\super{y,i}} = 1$ satisfied by the vector realization when $f(x) \neq f(y)$). The proof will conclude by using general results about operators that are a \emph{product of two reflections} (as is~$R_x$). If $\Pi$ and $\Delta$ were rank-one projectors, then the eigenvalues $e^{\pm i \theta}$ of the rotation $(2\Pi-\id)(2\Delta-\id)$ would be dictated by the angle~$\theta$ between the two projectors. Moreover, the norm of a state $\ket{t}$ orthogonal to $\Delta$ would decrease by a factor~$\sin(\theta/2) \leq \theta/2$ when projected on $\Pi$, as shown in \Cref{Fig:gap}.

  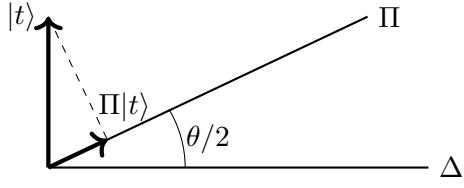
\begin{figure}
    \centering
    \begin{tikzpicture}
      \draw[thick] (0,0) -- (5,0) node[right] {$\Delta$};
      \draw[thick] (0,0) -- (4.2,2) node[right] {$\Pi$};

      \draw[ultra thick,->] (0,0) -- (0,2) node[left] {$\ket{t}$};
      \draw[ultra thick,->] (0,0) -- ($(0,0)!(0,2)!(4.2,2)$) ;
      \node at (1,0.8) {$\Pi \ket{t}$};
      \draw[dashed] (0,2) -- ($(0,0)!(0,2)!(4.2,2)$);

      \draw (1.8,0) arc[start angle=0,end angle=31,radius=1.5];
      \node at (2.1,0.38) {$\theta/2$};
    \end{tikzpicture}
    \caption{The orthogonal projection of the vector~$\ket{t}$ on~$\Pi$ has norm~$\norm{\Pi \ket{t}} = \sin(\theta/2) \norm{\ket{t}}$.}
    \label{Fig:gap}
  \end{figure}

  For higher-rank projectors, a similar argument applies to the norm of a state orthogonal to~$\Delta$ that is projected on the small-phase eigenspaces.

  \begin{lemma}[Effective spectral gap lemma \cite{LMR11c}]
    Let $\Pi$ and $\Delta$ be two projectors and set~$R = (2\Pi-\id)(2\Delta-\id)$. Let~$\Lambda_{\theta}$ be the projector onto the eigenspaces of~$R$ with eigenvalues in~$\set{e^{i\varphi} : \abs{\varphi} \leq \theta}$. Then, for any state~$\ket{t}$ with~$\Delta\ket{t} = 0$, we have~$\norm{\Lambda_{\theta} \Pi \ket{t}} \leq \frac{\theta}{2}\norm{\ket{t}}$.
  \end{lemma}

  By applying this lemma to the state $\ket{t_x^-}$, we conclude $\norm{\Lambda_{x,1/3T} \pss}^2 = \norm{\Lambda_{x,1/3T}  \allowbreak\Pi_x \ket{t_x^-}}^2 \leq \frac{1}{36T^2}\norm{\ket{t_x^-}}^2 = \frac{1}{36T^2}\pt{1 + 4T\sum_{i \in \ind} \norm{\ket{w\super{x,i}}}^2} \leq \frac{1}{36T^2} \pt{1 + 4T^2} \leq \frac{2}{9}$.
\end{proof}

The last ingredient of the construction is the well-known quantum phase estimation algorithm, which allows estimating (in superposition) the eigenphase of a guiding eigenvector.

\begin{lemma}[Phase estimation]
  \label{Lem:qpe}
  Let~$R$ be a unitary operator and~$\ket{\psi}$ be an eigenvector with eigenphase~$\varphi \in (-\pi,\pi]$, i.e., $R\ket{\psi} = e^{i\varphi} \ket{\psi}$. Given a precision parameter~$\eps \in (0,1)$, the quantum phase estimation algorithm implements a unitary $\mathrm{QPE}_R$ that uses the controlled-$R$ operation~$\bo{1/\eps}$ times and computes a superposition~$\mathrm{QPE}_R \pt{\ket{\psi}\otimes \ket{0}}= \ket{\psi}\otimes\pt{\sum_{\td{\varphi}} \alpha_{\td{\varphi}} \ket{\td{\varphi}}}$ of phase estimates~$\td{\varphi}$ such that the probability of measuring an estimate with error $\abs{\td{\varphi} - \varphi} < \eps$ is at least~$\sum_{\td{\varphi} : \abs{\td{\varphi} - \varphi} \leq \eps} \abs{\alpha_{\td{\varphi}}}^2 \geq 8/9$.
\end{lemma}

We now state the main theorem about converting the vector realization of a dual solution into a quantum algorithm.

\begin{theorem}[Dual algorithm]
  \label{Thm:dual}
  Let $\rf$ be a Boolean function and $\set{w\super{x,i}}_{x \in \rn^n, i \in \ind}$ be a set of complex vectors $w\super{x,i} \in \C^d$, for some integer~$d$, satisfying the condition
    \[\sum_{i : x_i \neq y_i} \ip{w\super{x,i}}{w\super{y,i}} = 1 \quad \text{for all $x,y$ such that $f(x) \neq f(y)$.}\]
  Then, there exists a quantum algorithm that computes~$f$ with query complexity,
    \[T = \bo[\bigg]{\max_{x \in \rn^n}  \sum_{i \in \ind} \norm{w\super{x,i}}^2}.\]
  Moreover, there exists at least one such set of vectors that satisfies~$T = \ta{\adv(f)}$.
\end{theorem}

\begin{proof}
  The algorithm simply consists of using phase estimation (\Cref{Lem:qpe}) on the unitary~$R_x$ and guiding state~$\pss$ (\Cref{Def:refl}) with precision~$\eps = 1/(6T)$, and measuring a value~$\td{\varphi}$ in the phase estimate register. If $\abs{\td{\varphi}} \leq 1/(6T)$ then it outputs~$1$, otherwise it outputs~$0$.

  The query complexity is~$\bo{T}$ by \Cref{Lem:qimplem,Lem:qpe}. It remains to show that the output is equal to~$f(x)$ with probability at least~$2/3$.

  If~$f(x) = 0$, the probability of measuring~$\abs{\td{\varphi}} \leq 1/(6T)$ is at least $\norm{\Lambda_{x,0} \pss}^2 \cdot 8/9 \geq 3/4 \cdot 8/9 = 2/3$, which is the squared norm of the component in $\pss$ with eigenphase~$0$ (\Cref{Lem:epGap}) multiplied by the probability of measuring an estimate less than $1/(6T)$ for such a component (\Cref{Lem:qpe}).

  If~$f(x) = 1$, the probability of measuring~$\abs{\td{\varphi}} \leq 1/(6T)$ is at most $\norm{\Lambda_{x,\frac{1}{3T}} \pss}^2 + 1/9 \leq 2/9 + 1/9 = 1/3$, which is the squared norm of the component in $\pss$ with eigenphase at most~$1/(3T)$ (\Cref{Lem:epGap}), added to the probability of measuring an estimate less than~$1/(6T)$ for the component with eigenphase at least~$1/(3T)$ (\Cref{Lem:qpe}).

  Finally, the algorithm can be made to work with the optimal value~$T = \ta{\adv(f)}$ by using a vector realization (\Cref{Lem:rep}) for an optimal solution to the dual from \Cref{Prop:dual}.
\end{proof}

Together with \Cref{Cor:adversary}, this result shows that the adversary value is a tight characterization of the quantum query complexity, up to constant factors.

\begin{corollary}
  \label{Cor:adversaryTight}
  There exist two universal constant $c_1 < c_2$ such that the quantum query complexity of any function $\rf$ satisfies $c_1 \adv(f) \leq Q(f) \leq c_2 \adv(f)$.
\end{corollary}


\subsection{Applications}

The dual of the adversary method has proven surprisingly useful in the design and analysis of new quantum algorithms. Several algorithmic frameworks have been developed based on the result established in \Cref{Thm:dual}, such as span programs~\cite{Rei09c}, learning graphs~\cite{Bel12c}, and, more recently, transducers~\cite{BJY24j}. These frameworks can simplify the search for dual solutions and the study of their properties.

In this section, we explicitly design the vector realizations of dual solutions that lead to optimal quantum algorithms for the \orf\ and \conf\ functions. We also use the composition property of the adversary value for solving the \aotf\ problem optimally.


\fakeparagraph{Application 1: The \orf\ function.}
We describe an optimal solution to the dual adversary proving that the lower bound~$Q(\orf) = \om{\sqrt{n}}$ is indeed optimal. We proceed by exhibiting a vector realization $w^{(x, i)}$ satisfying the conditions stated in \Cref{Lem:rep} with value $\max_{x \in \rn^n} \sum_{i \in \ind} \norm{w\super{x,i}}^2 = \bo{\sqrt{n}}$.

For convenience in the proof, we first establish a general result stating that the value of a vector realization can be balanced between the $0$-inputs and $1$-inputs, such that the maximum of $\max_{x \in \rn^n} \sum_{i \in \ind} \norm{w\super{x,i}}^2$ can be replaced with the geometric mean of the maximum evaluated on the $0$-inputs and $1$-inputs separately.

\begin{lemma}[Balanced vector realization]
  \label{Lem:rebalance}
  Suppose that $\set{\ket{w\super{x,i}}}$ is a vector realization of a solution with value $T = \max_{x \in \rn^n} \sum_{i \in \ind} \norm{w\super{x,i}}^2$, as described in \Cref{Lem:rep}. Define,
    \[T_0 = \max_{x : f(x) = 0} \sum_{i \in \ind} \norm{w\super{x,i}}^2 \quad \text{and} \quad T_1 = \max_{x : f(x) = 1} \sum_{i \in \ind} \norm{w\super{x,i}}^2.\]
  Then there exists a vector realization of a solution with value $\sqrt{T_0T_1}$.
\end{lemma}

\begin{proof}
  It suffices to define the vectors $v^{(x, i)} = (T_1/T_0)^{1/4}\cdot w^{(x, i)}$ when $f(x) = 0$ and $v^{(x, i)} = (T_0/T_1)^{1/4}\cdot w^{(x, i)}$ when $f(x) = 1$. When $f(x) \neq f(y)$, the condition $\sum_{i : x_i \neq y_i} \ip{v\super{x,i}}{v\super{y,i}} = 1$ is satisfied since the factors are cancelling out: $(T_1/T_0)^{1/4} (T_0/T_1)^{1/4} = 1$. This new solution has value $\max_{x \in \rn^n} \sum_i \norm{v\super{x,i}}^2 \leq \max\set{\sqrt{T_1/T_0} \cdot T_0, \sqrt{T_0/T_1} \cdot T_1} = \sqrt{T_0T_1}$.
\end{proof}

We now describe the algorithm for \orf.

\begin{proposition}[Dual algorithm applied to \orf]
  \label{Prop:orUB}
  There exists a quantum algorithm for the \orf\ function with query complexity~$Q(\orf) = \bo{\sqrt{n}}$.
\end{proposition}

\begin{proof}
  We describe a feasible vector realization with dimension~$d = 1$ (i.e., the vectors are scalar numbers). For each~$x \in \rn^n$ and~$i \in \ind$, set:
  \[w\super{x,i} =
      \left\{\begin{array}{c@{\hspace{4mm}}l}
         1 & \text{if $x$ is the all-$0$ input,} \\ [1mm]
         1 & \text{if $x_i = 1$ and $x_j = 0$ for all $j < i$,} \\[1mm]
         0 & \text{otherwise.}
       \end{array}
     \right.\]
  The condition $\sum_{i : x_i \neq y_i} \ip{w\super{x,i}}{w\super{y,i}} = 1$ is easily verified when~$\orf(x) \neq \orf(y)$, hence it is a valid vector realization.

  The all-$0$ input gives the value of $T_0 = \max_{x : f(x) = 0} \sum_{i \in \ind} \norm{w\super{x,i}}^2 = n$, and the other inputs gives the value of $T_1 = \max_{x : f(x) = 1} \sum_{i \in \ind} \norm{w\super{x,i}}^2 = 1$. Hence, by rebalancing the vector realization using \Cref{Lem:rebalance} (which amounts to multiplying $w\super{x,i}$ with $1/n^{1/4}$ when $x$ is the all-$0$ input, and $n^{1/4}$ otherwise) and applying \Cref{Thm:dual}, we obtain a quantum algorithm with query complexity~$\bo{\sqrt{T_0T_1}} = \bo{\sqrt{n}}$.
\end{proof}


\fakeparagraph{Application 2: \aotf.}
A \emph{read-once formula} is a Boolean function that can be represented as the evaluation of a rooted tree whose nodes are labeled with \orf, \andf\ and \notf\ gates, and where each input bit appears exactly once in the leaves. The composition property of the adversary value (\Cref{Prop:composition}) and \Cref{Thm:dual} enable the design of optimal quantum algorithms for such functions, achieving a query complexity of~$\ta{\sqrt{n}}$~\cite{Rei11c}. We outline the argument for the balanced \aotf, a depth-$2$ read-once formula with an \andf\ gate at the root and \orf\ gates at the first level (each with input size $\sqrt{n}$).

\begin{proposition}[Dual algorithm applied to \aotf]
  \label{Prop:aot}
  The quantum query complexity of the balanced \aotf\ function is~$Q(\andf \bullet \orf) = \ta{\sqrt{n}}$.
\end{proposition}

\begin{proof}
  The query complexities of the \orf\ and \andf\ functions over~$m = \sqrt{n}$ bits are $Q(\orf) = Q(\andf) = \ta{\sqrt{m}}$, as established, for instance, in \Cref{Prop:advOR,Prop:orUB} (it is easy to see that \orf\ and \andf\ must have exactly the same query complexity). By the composition property of the adversary value (\Cref{Prop:composition}) and its characterization of quantum query complexity (\Cref{Cor:adversaryTight}), we deduce that the complexity of the \aotf\ function is~$\ta{\sqrt{m} \times \sqrt{m}} = \ta{\sqrt{n}}$.
\end{proof}

While the lower bound may not seem particularly surprising here, the upper bound falls below the natural complexity~$\bo{\sqrt{n} \log n}$, which would be achieved using standard error reduction techniques when composing bounded-error algorithms. This errorless composition property is a striking feature of quantum query algorithms, allowing arbitrarily many levels of composition without any drift in the query complexity.


\fakeparagraph{Application 3: The \conf\ function.}
We return to the \conf\ problem, for which a lower bound of~$\om{n^{3/2}}$ was established in \Cref{Prop:conLB}. We complement this result with a matching upper bound due to Belovs and Reichardt~\cite{BR12c}, obtained by exhibiting the vector realization of an optimal solution to the dual adversary. These vectors will have a very compact description of dimension~$d = n^2$, leading to a quantum algorithm with only~$\bo{\log n}$ qubits of memory. This uses exponentially less space that an older quantum algorithm for \conf~\cite{DHHM06j} that requires $\bo{n \log n}$ memory.

\begin{proposition}[Dual algorithm applied to \conf]
  \label{Prop:conUB}
  \switchAMS{
  There exists a quantum algorithm for \conf\ with query complexity~$Q(\conf) = \bo{n^{3/2}}$.
  }{
  There exists a quantum algorithm for the \conf\ function with query complexity~$Q(\conf) = \bo{n^{3/2}}$.
  }
\end{proposition}

\begin{proof}
  We first describe a vector realization for the problem of deciding if two vertices~$s$ and~$t$ are connected by a path (called $st$-\conf), and we adapt it next to \conf.

  Recall that a graph over $n$ vertices is represented by its adjacency matrix $x \in \rn^{\binom{n}{2}}$. We let $C_x(v) \subseteq \ind$ denote the set of vertices that belong to the same connected component as the vertex~$v$. Given two vertices $s,t \in \rn^n$, we partition the set of graphs into two parts $\mathcal{G}_0^{st} \cup \mathcal{G}_1^{st} = \rn^{\binom{n}{2}}$, where~$\mathcal{G}_0^{st}$ contains the graphs that are not $st$-connected, and $\mathcal{G}_1^{st}$ contains the graphs that are $st$-connected (i.e., $t \in C_x(s)$). For each input~$x \in \rn^{\binom{n}{2}}$ and edge $\set{i,j} \subset \set{1,\dots,n}$ (representing a query index), we define the vector $\ket{v_{st}\super{x,\set{i,j}}} \in \spn\set{\ket{1},\dots,\ket{n}}$ as,
    \begin{itemize}
      \item If $x \in \mathcal{G}_0^{st}$ then:
      \[\ket{v_{st}^{(x,\set{i,j})}} = \left\{
        \begin{array}{ll}
          \ket{i} - \ket{j} & \text{if $i \in C_x(s)$ and $j \notin C_x(s)$,} \\
          0 & \text{otherwise.}
        \end{array}
      \right.\]
      \item If $x \in \mathcal{G}_1^{st}$ then fix a path \emph{of shortest length} from $s$ to $t$, and define:
      \[\ket{v_{st}^{(x,\set{i,j})}} = \left\{
        \begin{array}{ll}
          0 & \text{if $\set{i,j}$ is not an edge on that path,} \\
          \ket{i} & \text{if $\set{i,j}$ is an edge on the path and $i$ is visited first.} 
        \end{array}
      \right.\]
    \end{itemize}
  We claim that this construction is a valid vector realization for the problem of deciding whether the vertices $s$ and $t$ are connected. Indeed, let~$x \in \mathcal{G}_0^{st}$ and $y \in \mathcal{G}_1^{st}$. Then, the quantity $\sum_{\set{i,j} : x_{\set{i,j}} \neq y_{\set{i,j}}} \ip{v_{st}\super{x,\set{i,j}}}{v_{st}\super{y,\set{i,j}}}$ counts the number of times the (oriented) $st$-path associated with $y$ is leaving the connected component of $s$ in $x$, minus the number of times it is entering this component. The former occurs exactly one more time than the latter (since the path starts at $s \in C_x(s)$ and ends at $t \notin C_x(s)$), hence $\sum_{\set{i,j} : x_{\set{i,j}} \neq y_{\set{i,j}}} \ip{v_{st}\super{x,\set{i,j}}}{v_{st}\super{y,\set{i,j}}} = 1$.

  We now adapt this vector realization to the \conf\ problem. We use the basic observation that a graph is connected if and only if it is $st$-connected for $s = 1$ and all $t \in \set{2,\dots,n}$. Let $\mathcal{G}_0 \subset \rn^{\binom{n}{2}}$ be the set of all graphs that are not connected, and $\mathcal{G}_1$ those that are connected. For each~$x \in \rn^{\binom{n}{2}}$ and~$\set{i,j} \subset \set{1,\dots,n}$, we define $\ket{w\super{x,\set{i,j}}} \in \spn\set{\ket{k} \ket{t} : k \in \ind, t \in \set{2,\dots,n}}$ as,
    \begin{itemize}
      \item If $x \in \mathcal{G}_0$ then
      $\ket{w^{(x,\set{i,j})}} = \frac{1}{n - \abs{C_x(1)}} \sum_{t \notin C_x(1)} \ket{v_{1t}^{(x,\set{i,j})}}\ket{t}$.
      \item If $x \in \mathcal{G}_1$ then
      $\ket{w^{(x,\set{i,j})}} = \sum_{t \in \set{2,\dots,n}} \ket{v_{1t}^{(x,\set{i,j})}}\ket{t}$.
    \end{itemize}
  Given $x \in \mathcal{G}_0$ and $y \in \mathcal{G}_1$, we have $\sum_{\set{i,j} : x_{\set{i,j}} \neq y_{\set{i,j}}} \ip{w\super{x,\set{i,j}}}{w\super{y,\set{i,j}}} = \sum_{t \notin C_x(1)} \frac{1}{n - \abs{C_x(1)}} \cdot 1 = 1$, hence it is a valid vector realization for the \conf\ problem.

  Finally we show that the $0$-inputs have a value of at most $T_0 = 2(n-1)$ and the $1$-inputs have a value of at most $T_1 = (n-1)(n-2)$.
  Indeed, for all $x \in \mathcal{G}_0$, we have $\sum_{\set{i,j}} \norm{\ket{w\super{x,\set{i,j}}}}^2 = \frac{1}{(n - \abs{C_x(1)})^2} \sum_{t \notin C_x(1)} \sum_{\set{i,j}} \norm{\ket{v_{1t}^{(x,\set{i,j})}}}^2 = \frac{1}{(n - \abs{C_x(1)})^2} \sum_{t \notin C_x(1)} 2 \abs{C_x(1)} (n - \abs{C_x(1)}) = 2 \abs{C_x(1)} \allowbreak \leq 2(n-1)$. For all $x \in \mathcal{G}_1$, we have $\sum_{\set{i,j}} \norm{w\super{x,\set{i,j}}}^2 = \sum_{t \in \set{2,\dots,n}} \sum_{\set{i,j}} \norm{\ket{v_{1t}^{(x,\set{i,j})}}}^2 \leq \sum_{t \in \set{2,\dots,n}} (n-1) = (n-1)(n-2)$, where the inequality uses the fact that the shortest path between $s = 1$ and $t$ must be of size at most~$n-1$. By \Cref{Lem:rebalance} and \Cref{Thm:dual}, we can convert this vector realization into a quantum algorithm with query complexity~$\bo{\sqrt{T_0T_1}} = \bo{n^{3/2}}$.
\end{proof}